\documentclass{llncs}

\usepackage[vlined,linesnumbered]{algorithm2e}
\usepackage{algorithmic}
\usepackage{amsmath,amssymb%,amsthm
}
\usepackage{xspace}
\usepackage{dsfont}
\usepackage[colorlinks=true,urlcolor=webblue,linkcolor=webgreen,filecolor=webblue,citecolor=webgreen,pdfpagemode=UseOutlines,pdfstartview=FitH,pdfpagelayout=OneColumn,bookmarks=true]{hyperref}
\usepackage{enumerate}
\usepackage{graphicx}
\usepackage{enumitem}
\usepackage{xcolor}
\usepackage{authblk}
\usepackage{fullpage}
\usepackage{cleveref}
\def\bool{\{0,1\}}
\def\A{{\algo A}}
\def\B{\mathcal{B}}
\def\D{\mathcal{D}}
\def\O{\mathcal{O}}
\def\bydef{\stackrel{{\rm def}}{=}}
\def\xor{\oplus}

\definecolor{mgreen}{rgb}{.1,.7,0}

\usepackage{%algorithm, algpseudocode, 
	float}
\floatname{algorithm}{Procedure}

\usepackage[n,keys,asymptotics]{cryptocode}

\newcommand{\ignore}[1]{}

\newcommand{\Z}{\mathbb{Z}}
\newcommand{\N}{\mathbb{N}}

\newcommand{\ket}[1]{|#1\rangle} %
\newcommand{\bra}[1]{\langle #1|} %
\newcommand{\bad}{\textsf{BAD}}

\definecolor{webgreen}{rgb}{0,.5,0}
\definecolor{webblue}{rgb}{0,0,.5}
\newcommand\algo{\mathcal}

\newcommand{\Enc}{\ensuremath{\mathsf{Enc}}\xspace}

\newcommand{\expref}[2]{\texorpdfstring{\hyperref[#2]{#1~\ref{#2}}}{#1~\ref{#2}}}

\newcommand{\from}{\leftarrow}

\DeclareMathOperator*{\Exp}{\mathbb{E}}

\newcommand{\proj}[1]{\ensuremath{|#1\rangle \langle #1|}}

\newcommand{\braket}[2]{\left\langle #1 \mid #2 \right\rangle}

\newcommand{\Tr}{\mathrm{Tr}}
\newcommand{\CNOT}{\mathrm{CNOT}}

\newcommand{\formerqQ}{q_P}
\newcommand{\formerqC}{q_E}

\newcommand{\microspace}{\mspace{.5mu}} %
\renewcommand{\ket}[1]{\ensuremath{\lvert\microspace #1
		\microspace\rangle}} %
\renewcommand{\bra}[1]{\ensuremath{\langle\microspace #1
		\microspace\rvert}} %
\newcommand{\Oinv}{O^{\mathrm{inv}}}
\newcommand{\Per}{\mathrm{Per}}

\newcommand{\extraexplain}[1]{#1}
\newcommand{\Hyb}{\operatorname{\mathbf{H}}}
\newcommand{\swap}{{\sf swap}}
\newcommand{\Perms}{\mathcal{P}}
\newcommand{\Funcs}{\mathcal{F}}
\newcommand{\tweak}{\mathcal{T}}
\newcommand{\mask}{\textsf{mask}}
\newcommand{\smem}{Simplified \textsf{MEM}\xspace}

\newif\ifsubmission
\submissiontrue
\ifsubmission
\newcommand{\ga}[1]{}
\newcommand{\cm}[1]{}
\newcommand{\jnote}[1]{}
\newcommand{\chen}[1]{}
\else
\newcommand{\ga}[1]{{\noindent \textcolor{purple}{\emph{(GA:  #1)}}}{}}
\newcommand{\cm}[1]{{\noindent \textcolor{mgreen}{\emph{(CM:  #1)}}}{}}
\newcommand{\jnote}[1]{{\color{blue} [Jon's note: #1]}}
\newcommand{\chen}[1]{{\noindent \textcolor{orange}{\emph{(Chen:  #1)}}}{}}
\fi

\title{Post-Quantum Security of the \\ Even-Mansour Cipher}
\author{\vspace{-.8cm}}
\institute{\vspace{-.8cm}}

\author{Gorjan Alagic\inst{1} \and Chen Bai\inst{2} \and Jonathan Katz\inst{3} \and Christian Majenz\inst{4} }

\institute{\small QuICS, University of Maryland, and NIST\and Dept.\ of Electrical and Computer Engineering, University of Maryland \and Dept.\ of Computer Science, University of Maryland \and Department of Applied Mathematics and Computer Science, Technical University of Denmark%\and Department of Computer Science and Engineering, University of Connecticut
}

\begin{document}

\maketitle
\setcounter{footnote}{0}

\begin{abstract}
The Even-Mansour cipher is a simple method for constructing a (keyed) pseudorandom permutation $E$ from a public random permutation~$P:\bool^n \rightarrow \bool^n$. 
%It is a core ingredient in a wide array of symmetric-key constructions, including several lightweight cryptosystems presently under consideration for standardization by NIST. %\cm{is it a problem here that none of them uses actual Even-Mansour? We could write "The cipher and its variants..." instead of "It...".}\ga{I think this is fine, as we call it an ``ingredient.'' Your version sounds good too, though.}
It is secure against classical attacks, with optimal attacks requiring $q_E$ queries to $E$ and $q_P$ queries to $P$ such that $q_E \cdot q_P \approx 2^n$. If the attacker is given \emph{quantum} access to both $E$ and $P$, however, the cipher is completely insecure, with attacks using $q_E, q_P = O(n)$  queries known.

\smallskip
In any plausible real-world setting, however, a quantum attacker would have only \emph{classical} access to the keyed permutation~$E$ implemented by honest parties, even while retaining quantum access to~$P$. Attacks in this setting with $q_E \cdot q_P^2 \approx 2^n$ are known, showing that security degrades as compared to the purely classical case, but leaving open the question as to whether the Even-Mansour cipher can still be proven secure in this natural, ``post-quantum'' setting.
%\jnote{Revisit once we determine our concrete security bound}
%
%In this work, we show that this construction is secure in a very natural model of post-quantum security: the quantum-accessible public random permutation model. 

\smallskip
We resolve this question, showing that any attack in that setting requires $q_E \cdot q^2_P + q_P \cdot q_E^2 \approx 2^n$. Our results apply to both the two-key and single-key variants of Even-Mansour. Along the way, we establish several generalizations of results from prior work on quantum-query lower bounds that may be of independent interest.
%Our results imply \jnote{something about Elephant?}
%We show that, for any adversary which distinguishes the Even-Mansour cipher from a random permutation in this model, the number of quantum and classical queries must satisfy $q_\text{Q}^2 \cdot q_\text{C}=\Omega(2^n)$.  \jnote{It seems to me that Theorem~2 gives a bound on the advantage of the form $2^{-n/2} \cdot (\formerqC\sqrt{\formerqQ} + \formerqQ\sqrt{\formerqC})$, which would imply $q_\text{Q}^2 \cdot q_\text{C}=\Omega(2^n)$ or $q_\text{Q} \cdot q^2_\text{C}=\Omega(2^n)$.}
%jnote{Is that queries to both oracles, or just the public random permutation?}\cm{The BHT algorithm allows finding a claw in $\formerqC$ classical queries to the cipher and $\formerqQ$ quantum queries to the public random permutation whenever $\formerqQ^2\formerqC=\Omega(2^n)$. I think our techniques should allow us to prove a trade-off curve of lower bounds, that should look similar} \jnote{I guess my point is that instead of talking about ``queries'' whenever we give bounds, we should always talk about both classical and quantum queries.} 
%This bound matches the best known algorithms for breaking the cipher: BHT and Offline Simon.
%\smallskip
\end{abstract}

%\setcounter{tocdepth}{3}
%\tableofcontents
%~\\ \cm{end temporary TOC for convenience}
%\newpage

%%%%%%%%%%%%%%%%%%%%%%%%%%%
\section{Introduction}
%%%%%%%%%%%%%%%%%%%%%%%%%%%
The Even-Mansour cipher~\cite{EM97} is a well-known approach for constructing a block cipher~$E$ from a public random permutation
\mbox{$P: \bool^n \rightarrow \bool^n$}.
The cipher $E: \bool^{2n} \times \bool^n \rightarrow \bool^n$ is defined as
\[
%E: \bool^{2n} \times \bool^n &\rightarrow \bool^n\\
E_{k_1, k_2}(x) = P(x \xor k_1) \xor k_2
\] where, at least in the original construction, $k_1, k_2$ are uniform and independent.
Security in the standard (classical) setting %is well understood, where one considers computationally unbounded adversaries with classical oracle access to both $P$ and $E_{k_1,k_2}$ (and their inverses), 
%Here, ``classical oracle access'' refers to the ability to evaluate, e.g., the map $x \mapsto P(x)$. 
is well understood~\cite{EM97,DKS12}: roughly, an unbounded attacker with access to $P$ and $P^{-1}$ cannot distinguish whether it is interacting with $E_{k_1, k_2}$ and $E_{k_1, k_2}^{-1}$ (for uniform~$k_1, k_2$) or $R$ and $R^{-1}$ (for an independent, random permutation~$R$) unless it makes $\approx 2^{n/2}$ queries to its oracles. A variant where $k_1$ is uniform and $k_2=k_1$ has the same security~\cite{DKS12}. These bounds are tight, and key-recovery attacks using $O(2^{n/2})$ queries are known~\cite{EM97,DKS12}.  

Unfortunately, the Even-Mansour construction is insecure against a fully quantum attack in which the attacker is given \emph{quantum} access to its oracles~\cite{KM12,KLLN16}. In such a setting, the adversary can evaluate the unitary operators
\begin{align*} %\label{eq:quantum-oracles}
U_P &: \ket{x}\ket{y} \mapsto \ket{x}\ket{y \oplus P(x)}\\
U_{E_{k_1, k_2}} &: \ket{x}\ket{y} \mapsto \ket{x}\ket{y \oplus E_{k_1, k_2}(x)} \nonumber
\end{align*}
(and the analogous unitaries for $P^{-1}$ and $E_{k_1, k_2}^{-1}$) on any quantum state it prepares. Here, Simon's algorithm~\cite{Simon97} can be applied  to $E_{k_1, k_2} \oplus P$ to give a key-recovery attack using only $O(n)$ queries. 

To place this seemingly devastating attack in context, it is worth recalling that the original motivation for considering unitary oracles of the form above in quantum-query complexity was that one can always transform a classical circuit for a function~$f$ into a reversible (and hence unitary) quantum circuit for~$U_f$. In a cryptographic context, it is thus reasonable (indeed, necessary) to consider adversaries that use $U_f$ whenever $f$ is a function whose circuit they know. On the other hand, if the circuit for $f$ is \emph{not} known to the adversary, then there is no mechanism by which it can implement~$U_f$ on its own. In particular, if $f$ involves a private key, then the only way an adversary could possibly obtain quantum access to~$f$ would be if there were an explicit interface granting such access. In most (if not all) real-world applications, however, the honest parties using the keyed function~$f$ would implement $f$ using a classical computer; even if they implement $f$ on a quantum computer, there is no reason for them to support anything but a classical interface to~$f$. In such cases, an adversary would have no way to evaluate the unitary operator corresponding to~$f$
%it is then implausible that the adversary could force such devices to behave like quantum computers. %---especially if the only interface with $f$ is via a classical communication channel.
%\footnote{This leaves a plausible setting where ``quantum access'' models may be applicable: the honest device is a quantum computer, the adversary is connected via a quantum channel, but some application necessitates the use of classical cryptography.}.

In most real-world applications of Even-Mansour, therefore, an attacker would have only \emph{classical} access to the keyed permutation~$E_{k_1, k_2}$ and its inverse, while retaining quantum access to $P$ and~$P^{-1}$. In particular, this seems to be the ``right'' attack model for most applications of the resulting block cipher, e.g., constructing a secure encryption scheme from the cipher using some mode of operation. The setting in which the attacker is given classical oracle access to keyed primitives but quantum access to public primitives is sometimes called the ``Q1 setting''~\cite{BHN+19}; we will refer to it simply as the \emph{post-quantum} setting. %, as these terms describe the same theoretical model

Security of the Even-Mansour cipher in this setting is currently unclear.
Kuwakado and Morii~\cite{KM12} showed a key-recovery attack on Even-Mansour in this setting  that requires only $\approx 2^{n/3}$ oracle queries, using the BHT collision-finding algorithm~\cite{BHT97}. Their attack uses exponential memory but this was improved in subsequent work~\cite{HS18,BHN+19}, culminating in an attack using the same number of queries but with polynomial memory complexity. While these results demonstrate that the Even-Mansour construction is \emph{quantitatively} less secure in the post-quantum setting than in the classical setting, and show post-quantum security in certain restricted settings, they do not answer the more important \emph{qualitative} question of whether the Even-Mansour construction remains secure as a block cipher in the post-quantum setting, or whether attacks using polynomially many queries might be possible.

Concurrent results of Jaeger et al.~\cite{JST21} imply security of a forward-only variant of the Even-Mansour construction, as well as for the full
Even-Mansour cipher against \emph{non-adaptive} adversaries who make all their classical queries before any quantum queries.
They explicitly leave open the question of proving adaptive security in the latter setting.

\subsection{Our Results}

%%%%%%%%%%%%%%%%%%%%%%%%%%%

%\noindent{\bf Main result: post-quantum Even-Mansour is secure.}
%%%
%
%\jnote{Revisit once we determine our concrete security bound}
As our main result, we prove a lower bound showing that   $\approx 2^{n/3}$ queries are \emph{necessary} for attacking the Even-Mansour cipher in the post-quantum setting. In more detail, if $q_P$ denotes the number of (quantum) queries to $P, P^{-1}$ and $q_E$ denotes the number of (classical) queries to $E_{k_1, k_2}, E_{k_1, k_2}^{-1}$, we show that any attack succeeding with %non-negligible
constant probability requires either $q^2_P \cdot q_E = \Omega( 2^n)$ or $q_P \cdot q_E^2 = \Omega( 2^{n})$. %\footnote{Succeeding with non-negligible probability requires $q^2_P \cdot q_E+q_P \cdot q^2_E=2^{n-O(\log n)}$.} 
(Equating $q_P$ and $q_E$ gives the claimed result.) %Up to a polynomial factor, the same bound holds for the variant with $k_1=k_2$. 
Formally:

\begin{theorem}\label{thm:intro-EM}
Let $\A$ be a quantum algorithm making $\formerqC$ classical queries to its first oracle (including forward and inverse queries) and $\formerqQ$ quantum queries to its second oracle (including forward and inverse queries.)  Then
\begin{eqnarray*}
\lefteqn{\left|\Pr_{k_1, k_2,P} \left[\A^{E_{k_1, k_2}, P}(1^n) = 1\right]
- \Pr_{R,P} \left[\A^{R, P}(1^n) = 1\right]\right|} \hspace*{1.5in} \\
& \leq & 10 \cdot 2^{-n/2} \cdot \left(\formerqC \sqrt{\formerqQ} + \formerqQ \sqrt{\formerqC}\right)\,,
\end{eqnarray*}
where $P, R$ are uniform $n$-bit permutations, and the marginal distributions of $k_1, k_2\in\bool^n$ are uniform.
\end{theorem}
%\ga{Here we switch from $q_P$ and $q_E$ to $q_Q$ and $q_C$. Maybe we should stick to just one choice throughout the paper? If so, I think I prefer $q_P$ and $q_E$.}\cm{changed it as you suggested, in a reversible way}
%The same bound also holds for the single-key variant where $k_1=k_2$. %~\cite{DKS12}. 
The above applies, in particular, to  the two-key and one-key variants of the cipher.
A simplified version of the proof works also for the case where $P$ is a random function, we consider the cipher $E_k(x)=P(x \oplus k)$ with $k$ uniform, and $\A$ is given forward-only access to both $P$ and~$E$.

Real-world attackers are usually assumed to make far fewer queries to keyed, ``online'' primitives than to public, ``offline'' primitives. (Indeed, while an offline query is just a local computation, an online query requires, e.g., causing an honest user to encrypt a certain message.) In such a regime, where $\formerqC\ll \formerqQ$, the bound on the adversary's advantage in \expref{Theorem}{thm:intro-EM} simplifies to $O(\formerqQ \sqrt{\formerqC}\big/2^{n/2})$. In that case $\formerqQ^2 \formerqC=\Omega(2^n)$ is necessary for constant success probability, which matches the BHT and offline Simon algorithms~\cite{KM12,BHN+19}.\footnote{While our bound is tight with respect to the number of queries, it is loose with regard to the attacker's advantage, as both the BHT and offline Simon algorithms achieve advantage $\Theta(\formerqQ^2\formerqC \big/2^{n})$. Reducing this gap is an interesting open question.}

%\ga{Should add a discussion of optimality, particularly for the relevant regime of $\formerqQ$ much larger than $\formerqC$.} \jnote{I'm not entirely convinced that is the (only) relevant regime.} \ga{Well, we should certainly at least explicitly compare our bound to the best-known attacks, which I guess are BHT and Offline Simon. I guess $q_E$ and $q_P$ are the same in those algorithms, and all classical queries are made first. Is it worth discussing what the specialization of our proof to that case yields?}\cm{I wrote a paragraph (above)}

\medskip\noindent{\bf Techniques and new technical results.}\label{sec:technical-intro}
%%%
%\cm{Rewritten seciton start}
Proving \expref{Theorem}{thm:intro-EM} required us to develop new techniques 
that we believe are interesting beyond our immediate application. We describe the main challenge and its resolution in what follows.

As we have already discussed, in the setting of post-quantum security adversaries may have a combination of classical and quantum oracles. 
In addition to the Even-Mansour setting \cite{JST21}, this is the case, in particular, when a post-quantum security notion that involves keyed oracles is analyzed in the quantum random oracle model (QROM), such as when analyzing the Fujisaki-Okamoto transform~\cite{TU16,HHK17,BHHHP19,Zhandry2019,KSSSS20,DFMS21} or the Fiat-Shamir transform~\cite{Unruh17,KLS18,GHHM20}. In general, dealing with a mix of quantum and classical oracles  presents a problem: quantum-query lower bounds typically begin by ``purifying'' the adversary and postponing all measurements to the end of its execution, but this does not work if the adversary may decide what query to make to a classical oracle (or even whether to query a oracle at all)  \emph{based on the outcome} of an intermediate measurement. 
The works cited above address this problem in various ways (e.g. by specializing to non-adaptive adversaries \cite{JST21}), but often do so
%e.g., by simulating the classical oracle %\footnote{The security reduction of Kuchta et al.~\cite{KSSSS20} is particularly dependent on first making the classical oracle quantum-accessible by simulating it (without using the secret key), as their measure-rewind-measure technique only works for unitary adversaries.} 
by relaxing the problem and allowing quantum access to \emph{all} oracles.
This is not an option for us if we wish to prove security, because the Even-Mansour cipher is known not to be secure when the adversary has quantum access to all its oracles! %To the best of our knowledge, the only case in the literature where the described problem is dealt with in a non-trivial way is the concurrent work \cite{JST21}, where forward-only Even-Mansour with a random function is analyzed using Zhandry's superposition oracle technique. 
%\jnote{What do we want to say here about the TCC21 paper? I think assuming that all classical queries are made before all quantum queries is another form of relaxation that allows addressing the problem, but what about their result that handles adaptive queries in the forward-only setting?} \cm{we might want to say something about their result for forward only EM with a function once we have figured out whether it is correct.}
The only other work we are aware of that solves this problem is the concurrent work \cite{JST21}. Here, the authors overcome the described barrier for the forward-only Even-Mansour (see \cref{app:forward-EM}) using Zhandry's compressed oracle technique (which is not available for inverse-accessible permutations). Like previous works they delay all measurements, enforcing the classical-query nature of the adversary in a different way. 

%In our case, in contrast, mixed classical/quantum access to different oracles appears inherent. (In particular, if we allow quantum access to all oracles in our case, an efficient attack becomes possible!) 
Instead, we deal with the problem by dividing the execution of an algorithm that has classical access to some oracle~$O_1$ and quantum access to another oracle~$O_2$ into \emph{stages}, where a stage corresponds to a period between classical queries to~$O_1$. We then analyze the algorithm stage-by-stage. 
In doing so, however, we introduce another problem: 
the adversary may adaptively choose the number of queries to~$O_2$ in each stage based on outcomes of intermediate measurements. While it is possible to upper bound the number of queries to $O_2$ in each stage by the number of queries made to $O_2$ overall, this will (in general) result in a very loose security bound. %\footnote{This loss occurs even if one uses Markov's inequality to analyze a truncated version of the algorithm that never exceeds the expected number of queries in each round by too much. \jnote{I'm not sure this comment makes sense here.}} 
To avoid such a loss, we extend the  ``blinding lemma'' of Alagic et al.~\cite{AMRS20} so that (in addition to 
%accommodating two-way accessible permutations (and 
some other generalizations) we obtain a bound in terms of the \emph{expected} number of queries made by a distinguisher.

\begin{lemma}[Arbitrary reprogramming, informal]\label{lem:ada-blind-Intro}
%Let $F$ and $G$ be functions with the same domain and range, drawn from a distribution such that $\Pr[F(x)=G(x)]\ge 1-\epsilon$ for all inputs $x$. Then
%\[\left|\Pr[\mbox{$\D^F$ outputs 1}] - \Pr[\mbox{$\D^G$ outputs 1}]\right| \leq 2q \cdot \sqrt{\epsilon}\]
Consider the following game played by a distinguisher $\D$ making at most $q$ queries in expectation.
\begin{description}
	\item[Phase 1:] $\D$ outputs a function $F$ and a randomized algorithm $\B$ that specifies how to reprogram~$F$.
	\item[Phase 2:]   Randomness $r$ is sampled and $\B(r)$ is run to reprogram~$F$,  giving~$F'$. A uniform $b \in \bool$ is chosen, and $\D$ receives oracle access to either~$F$ (if $b=0$) or $F'$ (if $b=1$).
	\item [Phase 3:] $\D$ loses access to its oracle and receives~$r$; $\D$ outputs a bit $b'$.
\end{description}
Then $\left|\Pr[\mbox{$\D$ outputs 1}\mid b=0] - \Pr[\mbox{$\D$ outputs 1}\mid b=1]\right| \leq 2q \cdot \sqrt{\epsilon}$, where $\epsilon$ is an upper bound on the probability that any given input is reprogrammed.
\end{lemma}
The name ``arbitrary reprogramming'' is motivated by the facts that $F$ is  arbitrary (and known), and the adversary can reprogram $F$ arbitrarily---so long as some bound on the probability of reprogramming each individual input exists.

We also extend the ``adaptive reprogramming lemma'' of Grilo et al.~\cite{GHHM20} to the case of two-way-accessible, random permutations:
\begin{lemma}[Resampling for permutations, informal]\label{lem:prp-arl-Intro} Consider the following game played by a distinguisher~$\D$ making at most $q$ queries.
\begin{description}
\item [Phase 1:] $\D$ makes at most $q$ (forward or inverse) quantum queries to a uniform permutation $P:\bool^n \rightarrow \bool^n$.
	\item [Phase 2:] A uniform $b \in \bool$ is chosen, and $\D$ is allowed to make an arbitrary number of queries to an oracle that is either equal to~$P$ (if $b=0$) or $P'$ (if $b=1$), where $P'$ is obtained from $P$ by swapping the output values at two uniform points (which are given to $\D$.) Then $\D$ outputs a bit~$b'$.
\end{description}
\nopagebreak[4] Then $\left|\Pr[\mbox{$\D$ outputs 1}\mid b=0] - \Pr[\mbox{$\D$ outputs 1}\mid b=1]\right| \leq 4 \sqrt{q}\cdot 2^{-n/2}$.
\end{lemma}

This is tight up to a constant factor (cf.\ \cite[Theorem~7]{GHHM20}).
The name ``resampling lemma'' is motivated by the fact that here reprogramming is restricted to resampling output values from the \emph{same} distribution used to initially sample outputs of~$P$. 
While \expref{Lemma}{lem:ada-blind-Intro} allows for more general resampling, \expref{Lemma}{lem:prp-arl-Intro} gives a bound that is independent of the number of queries $\D$ makes after the reprogramming occurs. 

\medskip\noindent{\bf 
Implications for a variant of the Hidden Shift problem.}
%%%
In the well-studied Hidden Shift problem~\cite{vDHI06}, one is asked to find an unknown shift~$s$ by querying an oracle for a (typically injective) function $f$ on a group $G$ and an oracle for its shift $f_s(x) = f(x \cdot s)$. If both oracles are classical, this problem has query complexity superpolynomial in $\log|G|$. If both oracles are quantum, then the query complexity is polynomial~\cite{EHK04} but the algorithmic difficulty appears to depend critically on the structure of $G$ (e.g., while $G=\Z_2^n$ is easy~\cite{Simon97}, $G=S_n$ appears to be intractable~\cite{AR16}).

The obvious connection between the Hidden Shift problem  and security of Even-Mansour in general groups has been considered before~\cite{AR16,H17,BNP18}. In our case, it leads us to define two natural variants of the Hidden Shift problem:
\begin{enumerate}
\item ``post-quantum'' Hidden Shift: the oracle for $f$ is quantum while the oracle for $f_s$ is classical;
\item ``two-sided'' Hidden Shift: in place of $f_s$, use $f_{s_1, s_2}(x) = f(x \cdot  s_1) \cdot s_2$; if $f$ is a permutation, grant access to $f^{-1}$ and $f_{s_1, s_2}^{-1}$ as well.
\end{enumerate}
These two variants can be considered jointly or separately and, for either variant, one can consider worst-case or average-case settings~\cite{AR16}.
%define decisional and random versions of the problem, in addition to the standard injective search version~\cite{AR16}.
Our main result implies:

%\vspace*{-2pt}
\begin{theorem}[informal] \label{thm:hidden-shift-intro}
Solving the post-quantum Hidden Shift problem on any group $G$ requires a number of queries that is superpolynomial in $\log|G|$. This holds for both the one-sided and two-sided versions of the problem, and for both the worst-case and the average-case settings.
\end{theorem}
%\vspace*{-2pt}

\expref{Theorem}{thm:hidden-shift-intro} follows from the proof of \expref{Theorem}{thm:intro-EM} via a few straightforward observations. First, an inspection of the proof shows that the particular structure of the underlying group (i.e., the XOR operation on~$\bool^n$) is not relevant; the proof works identically for any group, simply replacing~$2^n$ with $|G|$ in the bounds.
The two-sided case of \expref{Theorem}{thm:hidden-shift-intro} then follows almost immediately: worst-case search is at least as hard as average-case search, and average-case search is at least as hard as average-case decision, which is precisely \expref{Theorem}{thm:intro-EM} (with the appropriate underlying group). Finally, as noted earlier, an appropriate analogue of \expref{Theorem}{thm:intro-EM} also holds in the ``forward-only'' case where $E(x) = P(x \oplus k)$ and $P$ is a random function. This yields the one-sided case of \expref{Theorem}{thm:hidden-shift-intro}.

%\medskip\noindent{\bf 
\subsection{Paper Organization}
%%%
In \expref{Section}{sec:lemmas} we state the technical lemmas needed for our main result. %(\expref{Theorem}{thm:intro-EM}). %; in particular, we give the formal statements of Lemmas~\ref{lem:ada-blind-Intro} and~\ref{lem:prp-arl-Intro}. 
In \expref{Section}{sec:EM} we prove \expref{Theorem}{thm:intro-EM}, showing post-quantum security of the Even-Mansour cipher (both the two-key and one-key variants), based on the technical lemmas. In \expref{Section}{sec:lemma-proofs} we prove the technical lemmas themselves. Finally, in \expref{Appendix}{app:forward-EM}, we give a proof of post-quantum security for the one-key, ``forward-only'' variant of Even-Mansour. 
While this is a relatively straightforward adaptation of the proof of our main result, it does not follow directly from it; moreover, it is substantially simpler and so may serve as a good warm-up for the reader before tackling our main result.
%\ga{The above has been revised significantly. Thoughts? (I commented out the earlier comments because they no longer apply, but feel free to return them.}
%\jnote{I'm a little confused by this section. If $f$ is random, then this is exactly (forward-only) Even-Mansour and so I'm not sure what we are claiming that we have not already. If $f$ is some specific function, then is it clear our results apply?}

%\ga{In quantum information, the interest is typically in the worst-case hardness of Hidden Shift for injective functions. So I think the only technical thing to show here is that, if there exists an algorithm which solves (post-quantum) HS for all injective functions, then it also solves Random HS, and our main result gives a lower bound for that. This is direct for the two-sided case, and I guess just an application of the collision bound for the one-sided case. But no, big-picture, this is just pointing out a connection, not claiming any significant additional result.}

%%%%%%%%%%%%%%%%%%%%%%%%%%%
\section{Reprogramming Lemmas}\label{sec:lemmas}
%%%%%%%%%%%%%%%%%%%%%%%%%%%

In this section we collect some technical lemmas that we will need for the proof of \expref{Theorem}{thm:intro-EM}. We first discuss a particular extension of the ``blinding lemma'' of Alagic et al.~\cite[Theorem~11]{AMRS20}, which formalizes \expref{Lemma}{lem:ada-blind-Intro}. We then state a generalization of the ``reprogramming lemma'' of Grilo et al.~\cite{GHHM20}, which formalizes \expref{Lemma}{lem:prp-arl-Intro}. The complete proofs of these technical results are given in \expref{Section}{sec:lemma-proofs}.

%\jnote{Does it bother anyone else that the order of the lemmas here is the opposite of their order in the previous section?}

We frequently consider adversaries with quantum access to some function $f : \bool^n \rightarrow \bool^m$. This means the adversary is given access to a black-box gate implementing the $(n+m)$-qubit unitary operator $\ket{x}\ket{y} \mapsto \ket{x}\ket{y \oplus f(x)}$.

%\medskip\noindent
%{\bf Notation.}
%We use $\Perms_n$ for the set of all permutations on $n$-bit inputs, and 
%$\Funcs_n$ for the set of all functions from $n$-bit inputs to $n$-bit outputs. \jnote{We may not ever use the latter, depending on whether we keep the forward-only EM proof.}

\subsection{Arbitrary Reprogramming}
%%%%%%%%%%%%%%%%%%%%%%%%%%%

%\jnote{Why is this called a blinding lemma? It actually seems more like a form of reprogramming\ldots}
%\ga{I guess that's historical. We first proved this in the paper \cite{AMRS20} on MACs. In that paper, the function is relabeled so the whole set $B$ is mapped to the value $\bot$ (signifying a ``blinding'') and the adversary's goal is ultimately to produce an input-output pair from this blinded region. But yes, it is a form of reprogramming.} \jnote{I guess I am wondering whether we should give it another name.} \ga{That's fine with me. I do not think the name ``blinding lemma'' is in widespread use, so now is the right time. :) }\cm{agree}
%\jnote{Any suggestions? :) it seems it would be nice to give it a name that involves ``reprogramming'' while still distinguishing it from the reprogramming in the previous section}
%\jnote{Non-adaptive reprogramming?}

%Next, we consider an experiment that involves a different form of reprogramming.
% scenario that is similar in spirit to the experiment considered in the previous section, but involves several technical differences.
%Now, 
Consider a reprogramming experiment that proceeds as follows. First, a distinguisher $\D$ specifies an arbitrary function~$F$ along with a probabilistic algorithm~$\mathcal B$ which describes how to reprogram $F$. Specifically, the output of $\mathcal B$ is a set of points $B_1$ at which $F$ may be reprogrammed, along with the values the function should take at those potentially reprogrammed points. Then $\D$ is given quantum access to either~$F$ or the reprogrammed version of~$F$, and its goal is to determine which is the case. When $\D$ is done making its oracle queries, it is also given the randomness that was used to run~$\mathcal B$. 
Intuitively, the only way $\D$ can tell if its oracle has been reprogrammed is by querying with significant amplitude on some point in~$B_1$. We bound $\D$'s advantage in terms of the probability that any particular value lies in the set $B_1$ defined by~$\mathcal B$'s output.

By suitably modifying the proof of Alagic et al.~\cite[Theorem~11]{AMRS20}, one can show that the distinguishing probability of $\D$ in the game described above is at most $2q \cdot \sqrt{\epsilon}$, where $q$ is an upper bound on the number of oracle queries and
$\epsilon$ is an upper bound on the probability that any given input $x$ is reprogrammed (i.e., that $x \in B_1$). However, that result 
is only proved for distinguishers with a fixed upper bound on the number of queries they make. %that fix the number of queries they make \emph{independently} of their oracle. 
To obtain a tighter bound for our application, we need a version of the result for distinguishers that may \emph{adaptively} choose how many queries they make based on outcomes of intermediate measurements. 
We recover the aforementioned bound if we let~$q$ denote the number of queries made by $\D$ \emph{in expectation}. 

For a function $F:\bool^m \rightarrow \bool^n$ and a set $B \subset \bool^m \times \bool^n$ such that each $x \in \bool^m$ is the first element of at most one tuple in~$B$,   define
\[
F^{(B)}(x) := 
\begin{cases}
y &\text{if } (x, y) \in B\\
F(x) &\text{otherwise.}
\end{cases}
\]
We prove the following in \expref{Section}{sec:reprog-lemma}:

%\ga{Do we still need the below, given that we have the adaptive version which is more general? (Do not just remove it outright, since the adaptive lemma refers to some small steps in the proof here.)}\cm{I think it would be good to make the adaptive version stand-alone and remove this version.}

%\medskip\noindent{\bf Adaptive distinguishers.}
%%%%%%%%%%%%%%%%%%%%%%%%%%%
%

%An easy one can be obtained by Markov's inequality. Let $\mathcal D$ be a distinguisher that makes an expected number of $q$ queries to distinguish an oracle from its $\epsilon$-blinded variant. Then we can consider an approximation $\tilde{\mathcal D}$  that makes at most $\kappa\cdot q$ queries, aborting if $\mathcal D$ would have made more queries. Using Markov's inequality to bound the probability of aborting, and optimizing $\kappa$, yields a bound that is $O(\sqrt q\epsilon^{\frac 1 4})$. We can, however, directly prove a blinding lemma for adaptive distinguishers, obtaining the same bound as in \expref{Lemma}{lem:blind}.

%We are now ready to state and prove our generalized reprogramming lemma.
\begin{lemma}[Formal version of \expref{Lemma}{lem:ada-blind-Intro}]\label{lem:ada-blind}
%	Consider the following experiment:
%	\begin{description}
%		\item[Phase 1:] $\D$ outputs descriptions of a function $P: \bool^m \rightarrow \bool^n$ and a randomized algorithm $\mathcal B$ whose output is a set $B \subset \bool^m \times \bool^n$ of tuples where any $x \in \bool^m$ is the first element of at most one tuple in~$B$, and such that for all $x\in \bool^m$ it holds that $\Pr[B \leftarrow \B: x \in B_1]\le \epsilon$.
%		\item[Phase 2:] A uniform random tape $r$ is chosen, and $\mathcal{B}(r)$ is run to obtain~$B$. Define
%		\[P_1(x) = P_{B,0^m}(x) = \begin{cases}
%			y & \mbox{if $(x, y) \in B$}\\
%			P(x) & \mbox{otherwise},
%		\end{cases}		\]
%		and let $P_0 = P$.		A uniform bit~$b$ is chosen, and $\D$ is given quantum access to~$P_b$. 
%\item[Phase 3:] $\D$ loses access to $P_b$, and  receives the randomness $r$ used to invoke~$\algo B$. It then outputs a guess~$b'$. 
%	\end{description}
Let $\D$ be a distinguisher in the following game:
\begin{description}
	\item[Phase 1:]  $\D$ outputs descriptions of a function $F_0=F: \bool^m \rightarrow \bool^n$ and a randomized algorithm $\mathcal B$ whose output is a set $B \subset \bool^m \times \bool^n$ where each $x \in \bool^m$ is the first element of at most one tuple in~$B$. Let $B_1 = \{x \mid \exists y: (x, y) \in B\}$
	%\begin{equation}\label{eq:blinding-eps}
and $	\epsilon = \max_{x \in \bool^m}\left\{\Pr_{B \leftarrow \B}[x \in B_1]\right\}.$
%	\end{equation}
	
	\item[Phase 2:]  $\mathcal{B}$ is run to obtain~$B$. Let $F_1=F^{(B)}$.
	%		\[P_1(x) = P_{B,0^m}(x) = \begin{cases}
	%		y & \mbox{if $(x, y) \in B$}\\
	%		P(x) & \mbox{otherwise},
	%		\end{cases}
	%		\]
	%and let $P_0 = P$.
	A uniform bit~$b$ is chosen, and
	$\D$ is given quantum access to~$F_b$. 
	\item[Phase 3:]  $\D$ loses access to $F_b$, and  receives the randomness~$r$  used to invoke~$\algo B$ in phase~2. Then $\D$ outputs a guess~$b'$. 
\end{description}
	For any $\D$ making $q$ queries in expectation when its oracle is $F_0$, %to $F_b$ when $b=0$,
	%(where the expectation is taken over the randomness used to invoke $\mathcal{B}$ and the measurement randomness of~$\mathcal{D}$)
%\jnote{Is this what we need for the eventual EM proof?}\cm{yes, I think so} 
it holds that 
$$\left|\Pr[\mbox{$\D$ outputs 1} \mid b=1] - \Pr[\mbox{$\D$ outputs 1} \mid b=0]\right| \leq 2q \cdot \sqrt{\epsilon}\,.$$
\end{lemma}

%\chen{Should it be $F_b$ instead of $F_0$?}\cm{no, $F_0$ is correct.}

\subsection{Resampling}
%%%%%%%%%%%%%%%%%%%%%%%%%%%

Here, we consider the following experiment: first, a distinguisher $\D$ is given quantum access to an oracle for a random function~$F$; then, in the second stage, $F$ may be ``reprogrammed'' so its value on a single, uniform point~$s$ is changed to an independent, uniform value. Because the distribution of $F(s)$ is the same both before and after any reprogramming, we refer to this as ``resampling.'' The goal for $\D$ is to determine whether or not its oracle was resampled. Intuitively, the only way $\D$ can tell if this is the case---even if it is given $s$ and unbounded access to the oracle in the second stage---is if $\D$ happened to put a large amplitude on $s$ in some query to the oracle in the first stage.
The lemmas we state here formalize this intuition. 

We begin by establishing notation and recalling a result of Grilo et al.~\cite{GHHM20}. 
Given a function $F:\bool^m \rightarrow \bool^n$ and $s \in \bool^m$, $y \in \bool^n$, define the ``reprogrammed'' function $F_{s \mapsto y}:\bool^m \rightarrow \bool^n$ as
\begin{eqnarray*}
F_{s \mapsto y} (w) = 
\begin{cases}
y &\text{if }w = s\\
F(w) &\text{otherwise.}
\end{cases}
\end{eqnarray*}
%\jnote{Should we make the notation consistent with $F_B$ by defining this as $F_{ \{(s,y)\} }$?} \ga{I feel like this notation is more intuitive for the single-reprogrammed-point case. And maybe it's actually better if the notation is different for the two different kinds of lemmas?}
%
The following is a special case of \cite[Prop.~1]{GHHM20}: 
\begin{lemma}[Resampling for random functions]\label{lem:arl}
Let $\D$ be a distinguisher in the following game:
	\begin{description} 
		\item[Phase 1:] A uniform $F:\bool^m \rightarrow \bool^n$ is chosen, and $\D$ is given quantum access to~$F_0=F$. 
		%This phase ends when $\D$ triggers \textsf{\emph{Reprogram}}.
		\item[Phase 2:] Uniform $s \in \bool^m$, $y \in \bool^n$ are chosen, and we let $F_1=F_{s \mapsto y}$.  
		A uniform bit~$b$ is chosen, and
		$\D$ is given $s$ 
		and quantum access to~$F_b$. Then $\D$ outputs a guess~$b'$.
	\end{description}
	For any $\D$ making at most $q$ queries to $F_0$ in phase 1, it holds that
	$$\left|\Pr[\mbox{$\D$ outputs 1} \mid b=1] - \Pr[\mbox{$\D$ outputs 1} \mid b=0]\right| \leq 1.5 \sqrt{q/2^m}\,.$$
\end{lemma}

%The second fact is similar, but applies to the case of two-way-accessible permutations. We consider an adaptive reprogramming game $G$ played by a two-stage adversary $\mathcal A=(\adver_0, \adver_1)$. $A_0$ gets quantum oracle access to a random permutation $\pi$ and its inverse, and outputs a quantum register $E$. Now two uniform inputs $x_0,x_1\in \bool^n$ and a uniform bit $b\in\{0,1\}$ are sampled. Define $\pi_0=\pi$ and 
We extend the above to the case of two-way accessible, random \emph{permutations}. Now, a random permutation~$P:\bool^n \rightarrow \bool^n$ is chosen in the first phase; in the second phase, $P$ may be reprogrammed by swapping the outputs corresponding to two uniform inputs. For $a, b \in \bool^n$, let $\swap_{a, b}: \bool^n \rightarrow \bool^n$ be the permutation that maps $a \mapsto b$ and $b \mapsto a$ but is otherwise the identity. We prove the following in \expref{Section}{sec:perm-ARL}:

%which satisfies $\swap_{a, b}(a) = b$ and $\swap_{a, b}(b) = a$ but is otherwise the identity.
%\[
%\swap_{a, b} (w) = 
%\begin{cases}
%b &\text{if } w = a;\\
%a &\text{if } w = b;\\
%w &\text{otherwise.}
%\end{cases}
%\]
%\jnote{We can write $\swap_{a,b}$ instead of $\swap_{(a,b)}$} \ga{Agreed, changing it.}
%For a permutation $P: \bool^n \rightarrow \bool^n$ and $s_1, s_2 \in \bool^n$, define the permutation $P_{s_1 \leftrightarrow s_2}$ as 
%\[
%P_{s_1 \leftrightarrow s_2}(w)=\begin{cases}
%P(s_2) & \text{if }w=s_1\\
%P(s_1) & \text{if }w=s_2\\
%			\pi(x_{i\oplus 1})&\text{if }x=x_i\\
%P(w)& \text{otherwise.}
%\end{cases}
%\]
%\jnote{I do not think we use this notation ever again (e.g., when we apply this lemma in our EM proof).}
%Note that sampling two random inputs, two random outputs, or one random input and one random output, are all equivalent. 
%\jnote{I'm not quite sure what you mean by that.}\cm{the remark was imprecise and not exactly correct, so I removed it for now (We can recover and correct it if we need).} 
%\jnote{Is it essential that $\pi_1$ be a permutation? In principle, we could set $\pi_1=\pi_{s \rightarrow y}$ and define $\pi_1^{-1}$ in the natural way except that $\pi_1^{-1}(y)=s$ and $\pi_1^{-1}(\pi(s))=\perp$. Not sure if this simplifies anything for our purposes.} \ga{I think we could do this, but the messiness would just be transported somewhere else (like bounding collision-finding.)}
%Now $\adver_1$ gets input $(x_0, x_1, E)$ and oracle access to $\pi_b$ and outputs~$b'$. $\adver$ wins if $b=b'$.
%We claim the following.

\begin{lemma}[Formal version of \expref{Lemma}{lem:prp-arl-Intro}]\label{lem:prp-arl}\label{lem:prp-arl-proved}
Let $\D$ be a distinguisher in the following game:
	\begin{description} 
		\item[Phase 1:] A uniform permutation $P:\bool^n \rightarrow \bool^n$ is chosen, and $\D$ is given quantum access to~$P_0=P$ and $P^{-1}_0=P^{-1}$. %This phase ends when $\D$ triggers \textsf{\emph{Reprogram}}.
		\item[Phase 2:] Uniform $s_0, s_1 \in \bool^n$ are chosen, and we let $P_1=P\circ \swap_{s_0,s_1}$.  
		Uniform~$b \in \bool$ is chosen, and
		$\D$ is given $s_0, s_1$, 
		and quantum access to~$P_b, P^{-1}_b$. Then $\D$ outputs a guess~$b'$.
	\end{description}
	For any $\D$ making at most $q$ queries (combined) to $P_0, P^{-1}_0$ in the first phase, 
	$\left|\Pr[\mbox{$\D$ outputs 1} \mid b=1] - \Pr[\mbox{$\D$ outputs 1} \mid b=0]\right| \leq 4 \sqrt{q/2^n}$.
\end{lemma}

% as \cref{thm:prp-arl-proved}. 

%\jnote{We should either put all proofs in Section~2, or all proofs in later sections.} \ga{Agreed. Since the proof of the ARL is quite long, I went with putting both proofs in a later section.}
%\chen{I went  through lemma 2 to try to understand this, the result $2\sqrt{\frac{q}{2^n}}$ in equation (61) makes sense to me, but I do not understand why we have  $4\sqrt{\frac{q}{2^n}}$ here?} 
%\cm{The reason is that \cref{thm:prp-arl-proved} is formulated in terms of a game where the permutation undergoes resampling with probability 1/2, while we formulate the above lemma in terms of the distinguishing advantage, i.e. the difference of the probabilities of outputting 1 in the case of resampling and no resampling. It is a shor calculation to show that these statements are equivalent, but maybe we should still harmonize the statements.}\chen{Ah I see, that makes sense, thanks!}

%%%%%%%%%%%%%%%%%%%%%%%%%%%
\section{Post-Quantum Security of Even-Mansour}\label{sec:EM}
%%%%%%%%%%%%%%%%%%%%%%%%%%%

%\subsection{Notes}
%%%%%%%%%%%%%%%%%%%%%%%%%%%

%A non-exhaustive list of things that need to change:
%\begin{enumerate}
%\item hybrids will look quite similar, and will now run across all classical queries (forward and inverse.) These queries are symmetrical: each one simply reveals an input-output pair of $R$. We then have to modify the $P$-oracle accordingly in the later phases.
%\item $\Hyb_j$ to $\Hyb_j'$ will need a version of \expref{Lemma}{lem:arl} that handles two-sided permutation access. See the section below for one version of such a fact. Note this is not just a straightforward consequence of the results of~\cite{GHHM20}. 
%\item $\Hyb_j'$ to $\Hyb_{j+1}$ will now involve a two-sided blinding set: the table $T$ of $(x, y)$ pairs will get shifted to $\{ (x \oplus k_1, y \oplus k_2) : (x, y) \in T\}$. The argument of \expref{Lemma}{lem:blind} should not have to change too much to accomodate this.
%\end{enumerate}

%\subsection{The setup}
%%%%%%%%%%%%%%%%%%%%%%%%%%%

We now establish the post-quantum security of the Even-Mansour cipher based on the lemmas from the previous section. Recall that the Even-Mansour cipher is defined as $E_k(x) := P(x \oplus k_1) \oplus k_2$,
where 
$P : \bool^n \rightarrow \bool^n$ is a public random permutation and $k = (k_1, k_2) \in \bool^{2n}$ is a key.
Our proof assumes only that the marginal distributions of $k_1$ and $k_2$ are each uniform. 
%, i.e., that $\Pr[k_1 = k] = \Pr[k_2 = k] = 2^{-n}$ for any $k \in \bool^n$.
This covers the original Even-Mansour cipher~\cite{EM97} where $k$ is uniform over~$\bool^{2n}$ as well as the one-key variant~\cite{DKS12} where $k_1$ is uniform and then $k_2$ is set equal to~$k_1$.
%We consider both the original, two-key version~\cite{EM97} where $k\in\bool^{2n}$ is uniform, as well as the single-key variant~\cite{DKS12} where $k_1$ is uniform and $k_2=k_1$.  
%\jnote{A simpler construction chooses uniform $k_1 \in \bool^n$ and then sets $k_2=k_1$~\cite{DKS12}. Note that in both of these constructions the marginal distributions of $k_1$ and $k_2$ are uniform, i.e., $\Pr_k[k_1 = K] = \Pr_k[k_2 = K] = 1/2^n$ for all $K \in \bool^n$.}

For $E_k$ to be efficiently invertible, the permutation $P$ must itself support efficient inversion; that is, the oracle for $P$ must be accessible in both the forward and inverse directions. We thus consider adversaries $\A$ who can access both the cipher $E_k$ and the permutation $P$ in both the forward and inverse directions. The goal of $\A$ is to distinguish this world from the ideal world in which it interacts with independent random permutations~$R, P$.
In this section, it will be implicit in our notation that all oracles are two-way accessible. %\jnote{I'm not sure we need the $(\cdot)^\pm$ notation, since everything in this section is two-sided access.}\cm{I agree, at least as long as we do not start treating the forward and inverse oracles separately.} \ga{Agreed, dropped them.}

In the following, we let $\Perms_n$ be the set of all permutations of~$\bool^n$. We write $E_k[P]$ to denote the Even-Mansour cipher using permutation $P$ and key~$k$; we do this both to emphasize the dependence on $P$, and to enable references to Even-Mansour with a permutation other than~$P$. Our main result is as follows:

%\jnote{Theorem~1 is already fully formal; we could omit repeating the theorem below if we need space.}

\begin{theorem}[\expref{Theorem}{thm:intro-EM}, restated]\label{thm:full} 
Let	
$D$ be a distribution over $k=(k_1, k_2)$ such that the marginal distributions of $k_1$ and $k_2$ are each uniform, and let
$\A$ be an adversary making $\formerqC$ classical queries to its first oracle and $\formerqQ$ quantum queries to its second oracle. %\jnote{and let 
%$D$ be a distribution on 
%$(k_1, k_2) \in \bool^{2n}$
%such that the marginal distributions of $k_1$ and $k_2$ are uniform}
Then
	\begin{eqnarray*}
		\lefteqn{\left|\Pr_{\substack{k \leftarrow D \vspace*{1pt}\\P \from \Perms_n }} \left[\A^{E_k[P], P}(1^n) = 1\right]
			- \Pr_{R, P \from \Perms_n} \left[\A^{R, P}(1^n) = 1\right]\right|} \hspace*{1.5in} \\
		& \leq & 10 \cdot 2^{-n/2}\left(\formerqC\sqrt{\formerqQ}+\formerqQ\sqrt{\formerqC}\right).
	\end{eqnarray*}
%If instead, $k_1$ is uniform and $k_2=k_1$, then we obtain the bound
%	\begin{eqnarray*}
%	& \left(\frac 5 2+n\right)\cdot q_E^2 \cdot 2^{-n}+2 \cdot 2^{-n/2}\cdot \left(\formerqC\sqrt{\formerqQ}+\formerqQ\sqrt{\formerqC}\right).
%\end{eqnarray*}
\end{theorem}
%\chen{I guess it should be $k \from \bool^n \times \bool^n$ on the equation above.}\cm{$k$ is drawn from a more general distribution. I fixed it by removing the distribution altogether from the formula, it is mentioned right below how $k$ is sampled. Please check if you find it clear enough.}
%Abusing notation slightly, we will use $D(k_i|k_j)$ to denote both conditional distributions and let the index $j\in\{1,2\}$ indicate which of the two we mean. \ga{unclear, come back to this.}
%\chen{ The notation $D(k_i|k_j)$ seems a little bit weird to me, since later we set $k_2 \from D(\cdot |k_1)$, so how about just using $k_i \from D(\cdot|k_j)$ to denote the marginal distribution of $k_i$ given $k_j$, which is uniform?}\cm{no, the marginal is uniform, but not the conditional distribution. In particular, for the single-key variant of EM, $D(k_i|k_j)=\delta_{k_ik_j}$} \ga{The text these comments refer to was moved to the proof of Lemma 6 below, where it is actually used for the first time. Chen, if you're satisfied with this discussion, please \% all this out.}
\begin{proof}
Without loss of generality, we assume $\A$ never makes a redundant classical query; that is, once it learns an input/output pair $(x, y)$ by making a query to its classical oracle, it never again submits the query $x$ (respectively, $y$) to the forward  (respectively, inverse) direction of that oracle.

We divide an execution of $\A$ into $\formerqC+1$ stages $0, \ldots, \formerqC$, where the $j$th stage corresponds to the time between the $j$th and $(j+1)$st classical queries of~$\A$. In particular, the $0$th stage corresponds to the period of time before $\A$ makes its first classical query, and the $\formerqC$th stage corresponds to the period of time after $\A$ makes its last classical query. We allow $\A$ to adaptively distribute its $\formerqQ$ quantum queries between these stages arbitrarily. 
We let $q_{P,j}$ denote the expected number of queries $\A$ makes in the $j$th stage in the ideal world~$\A^{R,P}$; note that $\sum_{j=0}^{q_E} q_{P,j} = q_P$.

%Let $Q_j$ be the random variable indicating how many queries are made in the $j$th stage, and let $q_{P, j}$ denote the expectation of $Q_j$, in the ideal world. Then
%\begin{equation}\label{eq:query-sum}
%{\textstyle \sum_j q_{P, j} = \sum_j \Exp[Q_j] = \Exp\Bigl[\sum_j Q_j\Bigr]} = \Exp[\formerqQ] = \formerqQ.
%\end{equation}

%we will consider hybrids in which $\A$ is run with various sequences of oracles. Each query will now be made to some permutation, either in the forward or the inverse direction. Without loss of generality, such a sequence consists of $q$-many ``query rounds.'' in each round, $\A$ will make a single classical query to an oracle $C_1^{\pm}$ followed by a polynomial number of quantum queries to oracle $Q_1^{\pm}$. We will then denote an execution of $\A$ with such a query sequence by writing
%\begin{equation}
%C_1^{\pm}, Q_1^{\pm}; C_2^{\pm}, Q_2^{\pm}; \cdots ; C_q^{\pm}, \formerqQ^{\pm}\,.
%\end{equation}
%Note that the number of quantum queries in each round is suppressed in this notation.

Recall that $\swap_{a, b}$ swaps $a$ and $b$. Given a permutation $P$, an ordered list of pairs $T = \big((x_1, y_1), \dots, (x_t, y_t)\big)$, % where there are no repeated $x_i$ and no repeated $y_i$, 
and a key $k = (k_1, k_2)$, define
\begin{equation}\label{eq:PsubTk}
P_{T, k} = \swap_{P(x_1 \oplus k_1), y_1\oplus k_2} \circ \cdots \circ \swap_{P(x_t \oplus k_1), y_t \oplus k_2} \circ P\,.
\end{equation}
%\jnote{Given a permutation $P$, values $x, y$, % where there are no repeated $x_i$ and no repeated $y_i$, 	and a key $k = (k_1, k_2)$, define
%	\begin{equation}\label{eq:PsubTk}
%	P_{(x,y), k} = \swap_{P(x \oplus k_1), y \oplus k_2} \circ P\,.
%	\end{equation}}
(If $T$ is empty, then $P_{T, k} = P$.) 
Intuitively, assuming the $\{x_i\}$ are distinct and the $\{y_i\}$ are distinct, $P_{T, k}$ is a ``small'' modification of $P$ for which
$E_k[P_{T, k}](x_i) = y_i$ for all~$i$. (Note, however, that this may fail to hold if there is an ``internal collision,'' i.e., $P(x_i\xor k_1)=y_j \xor k_2$ for some $i \neq j$. But such collisions occur with low probability over choice of~$k_1, k_2$.) \jnote{I'm not sure if this note is helpful or more confusing.}

%Note that for a random permutation $P$,  $k$ sampled as prescribed by our variant of Even-Mansour, and $T$ as above such that $k$ is independent of $(P,T)$,
%\begin{equation}\label{eq:bad-reprogramming}
%\Pr\left[\exists (x, y) \in T: P(x \oplus k_1) \oplus k_2 = y\right]\le Ct(t-1)n 2^{-n}
%\end{equation}
%for some small constant $C$ by \cref{lem:reprogramming-works} and a union bound.
%\jnote{See my note by \cref{lem:reprogramming-works}}
%\chen{Is it $P$ instead of $P_{T,k}$ in equation (8)?}
 %\ga{Coming back to this: my comments here before were wrong, I think. I guess the issue is if $P : x_i \oplus k_1 \mapsto y_j \oplus k_2$ for some $i, j$. But actually, even in that case, if we do the $i$-th SWAP first, i.e.,
%$$\swap_{(P(x_j \oplus k_1), y_j \oplus k_2)} \circ \swap_{(P(x_i \oplus k_1), y_i \oplus k_2)} \circ P$$
%then I do not think there's a problem. 
%So then it seems the only issue is if $P : x_i \oplus k_1 \mapsto y_j \oplus k_2$ \textbf{AND} $P : x_j \oplus k_1 \mapsto y_i \oplus k_2$ for some $i, j$.} \jnote{It seems to be a problem if $P(x_i \xor k_1) = y_j \xor k_2$ for some $i>j$.}
%
%\paragraph{Defining the hybrids.}
%%%
%
%Let $D$ denote the distribution of keys on $\bool^n \times \bool^n$. 

We now define a sequence of experiments $\Hyb_j$, for $j=0, \ldots, \formerqC$. 

\medskip\noindent \textbf{Experiment} $\Hyb_j$. Sample $R,P \from \Perms_n$  and $k \leftarrow D$. Then:
\begin{enumerate} %[label=(\Roman*)]
\item Run $\A$, answering its classical  queries using $R$ and its quantum queries using $P$, stopping immediately \emph{before} its
$(j+1)$st classical query. Let $T_j = \big((x_1, y_1), \dots, (x_j, y_j)\big)$ be the ordered list of all input/output pairs that $\A$ received from its classical oracle.

\item For the remainder of the execution of $\A$, answer its classical queries using $E_k[P]$ and its quantum queries using $P_{T_j, k}$. 
\end{enumerate}
We can compactly represent $\Hyb_j$ as the experiment in which $\A$'s queries are answered using the oracle sequence
\[
\underbrace{P, R, P, \cdots, R, P,}_{\mbox{\scriptsize $j$ classical queries}} \, \underbrace{E_k[P], P_{T_j, k}, \cdots, E_k[P], P_{T_j, k}}_{\mbox{\scriptsize $\formerqC-j$ classical queries}}\,.
\]
%\jnote{I think we can instead define it as:
%	\[
%	\underbrace{P, R, P, \cdots, R, P,}_{\mbox{\scriptsize $j$ classical queries}} \, \underbrace{E_k[P], P, \cdots, E_k[P], P}_{\mbox{\scriptsize $\formerqC-j$ classical queries}}\,.
%	\]}
Each appearance of $R$ or $E_k[P]$ indicates a single classical query. Each appearance of $P$ or $P_{T_j, k}$ indicates a stage during which $\A$ makes multiple (quantum) queries to that oracle but no queries to its classical oracle. 
Observe that $\Hyb_0$ corresponds to the execution of $\A$ in the real world, i.e., $\A^{E_k[P], P}$, and that $\Hyb_{q_E}$ is the execution of $\A$ in the ideal world, i.e., $\A^{R, P}$.

For $j=0, \ldots, \formerqC-1$, we introduce additional experiments~$\Hyb_j'$:

\medskip\noindent \textbf{Experiment} $\Hyb_j'$. 
Sample $R, P \from \Perms_n$ and $k \leftarrow D$. Then:
\begin{enumerate} %[label=(\Roman*)]
\item Run $\A$, answering its classical queries using $R$ and its quantum queries using~$P$, stopping immediately \emph{after} its $(j+1)$st classical query. Let $T_{j+1} = \big((x_1, y_1), \dots, (x_{j+1}, y_{j+1})\big)$ be the ordered list of all input/output pairs that $\A$ learned from its classical oracle.

\item For the remainder of the execution of $\A$, answer its classical queries using $E_k[P]$ and its quantum queries using $P_{T_{j+1}, k}$.
\end{enumerate}
Thus, $\Hyb'_j$ corresponds to running $\A$ using the oracle sequence
\[
\underbrace{P, R, P, \cdots, R, P,}_{\mbox{\scriptsize $j$ classical queries}}\, R, P_{T_{j+1}, k},\, \underbrace{E_k[P], P_{T_{j+1}, k} \cdots, E_k[P], P_{T_{j+1}, k}}_{\mbox{\scriptsize $\formerqC-j-1$ classical queries}}\,.
\]
%\jnote{I think we can change this to:
%	\[
%	\underbrace{P, R, P, \cdots, R, P,}_{\mbox{\scriptsize $j$ classical queries}}\, R, P_{(x^*,y^*), k},\, \underbrace{E_k[P], P, \cdots, E_k[P], P}_{\mbox{\scriptsize $\formerqC-j-1$ classical queries}}\,,
%	\]
%where $x^*=x_{j+1}$ and $y^*=R(x^*)$.}
In Lemmas~\ref{lem:step-two-perps} and~\ref{lem:step-one-perps}, we establish bounds on the distinguishability of $\Hyb_j'$ and~$\Hyb_{j+1}$, as well as $\Hyb_j$ and $\Hyb_j'$. For $0 \leq j < \formerqC$ these give:
\begin{eqnarray*}
	\left|\Pr[\A(\Hyb_j') = 1] - \Pr[\A(\Hyb_{j+1})=1]\right|  & \leq &  2 \cdot q_{P,j+1} \cdot \sqrt{\frac{2\cdot (j+1)}{2^n}}. \label{eq:jprime-to-j} \\
%\mbox{for $0 \leq j \leq \formerqC$:}& \;\; &
\left|\Pr[\A(\Hyb_j) = 1] - \Pr[\A(\Hyb'_j)=1]\right|  & \leq &   8 \cdot \sqrt{\frac{\formerqQ}{2^n}}+2q_E\cdot2^{-n} \label{eq:j-to-jprime}
%\end{eqnarray}
%and 
%\begin{eqnarray}
%\mbox{for $0 \leq j < \formerqC$:} & \;\; &
\end{eqnarray*}
%\chen{Until here I think our $q_P$ denotes the number of all quantum queries, however, in lemma 6 $6q_P$ denotes the number of quantum queries made in the first phase, i.e, in equation (2) and (4) $q_P$ should be $\sum_{i=0}^{j}q_{P,i}$ with $q_{P,0}=0$, in this way our bound should be 
%\begin{align*}
%	&\left|\Pr_{\substack{P \from \Perms_n\\k \from \bool^{2n}}} \left[\A^{E_k[P], P}(1^n) = 1\right]
%		- \Pr_{\substack{P \from \Perms_n\\R \from \Perms_n}} \left[\A^{R, P}(1^n) = 1\right]\right| \\
%		& = \left|  \Pr[\A(\Hyb_0) = 1] -  \Pr[\A(\Hyb_{q_E})=1]\right| \\
%	& \leq  \sum_{j=0}^{\formerqC-1}\left(1.5 \cdot \sqrt{\frac{\sum_{i=0}^{j}q_{P,i}}{2^n}}+4j\cdot2^{-n}+2 q_{P,j+1} \sqrt{\frac{ j}{2^n}}\right)\\
%	&\leq 2 \formerqC^2\cdot 2^{-n}+ 2\cdot \sum_{j=0}^{\formerqC-1}\left(\sqrt{\frac{\formerqQ}{2^n}}+q_{P,j+1} \sqrt{\frac{\formerqC}{2^n}}\right)\\
%	&=2\formerqC^2\cdot 2^{-n}+2\cdot 2^{-n/2}\cdot \left(\formerqC\sqrt{\formerqQ}+\formerqQ\sqrt{\formerqC}\right).
%\end{align*}
%Looks like the bound is same as below, but I think we need to mark this difference.
%}\cm{The bound in \cref{lem:step-one-perps} is formulated in terms of $q_P$, the total number of quantum queries, i.e. in the proof we upper-bound  the number of quantum queries before reprogramming by the total number of quantum queries. This does not lose us anything, as an adversary could decide to first make all (or half, as suggested by Jon at some point) its quantum queries, and only then start to query $E$.}
Using the above, we have
\begin{align*}
%	&\left|\Pr_{\substack{k \from D \vspace*{1pt}\\P \from \Perms_n}} \left[\A^{E_k[P], P}(1^n) = 1\right]
%		- \Pr_{R,P \from \Perms_n} \left[\A^{R, P}(1^n) = 1\right]\right| \\
		& \left|  \Pr[\A(\Hyb_0) = 1] -  \Pr[\A(\Hyb_{q_E})=1]\right| \\
	& \leq  \sum_{j=0}^{\formerqC-1}\left( 8 \cdot \sqrt{\frac{\formerqQ}{2^n}}+2q_E\cdot2^{-n}+2 \cdot q_{P,j+1} \sqrt{\frac{2 \cdot (j+1)}{2^n}}\right)\\
	&\leq 2 \formerqC^2\cdot 2^{-n}+  \sum_{j=0}^{\formerqC-1}\left(8\cdot \sqrt{\frac{\formerqQ}{2^n}}+2\cdot q_{P,j+1} \sqrt{\frac{2\formerqC}{2^n}}\right)\\
	&\leq 2\formerqC^2\cdot 2^{-n}+2^{-n/2}\cdot \left(8\formerqC\sqrt{\formerqQ}+2\cdot \formerqQ\sqrt{2\formerqC}\right).
\end{align*}
%using the triangle inequality for the first inequality, the fact that $j < \formerqC$ in the next step, and the fact that $\sum_j q_{P, j} = \formerqQ$ in the last step. Repeating this chain of inequalities for the case $k_1=k_2$ (with Eq.~\eqref{eq:j-to-jprime} replaced by Eq.~\eqref{eq:arl-based-lem-2} from Lemma~\ref{lem:step-one-perps}), yields the claimed bound for that case.

We now simplify the bound further. If \mbox{$\formerqQ=0$}, then $E_k$ and $R$ are perfectly indistinguishable and the theorem holds;
%; in particular, for $\formerqQ=0$ the  theorem holds. 
thus, we may assume $\formerqQ\ge 1$. We can also assume $\formerqC < 2^{n/2}$ since otherwise the bound is larger than~1. %$2\cdot 2^{-n/2}\formerqC\sqrt{\formerqQ} \geq 1$ 
Under these assumptions, we have $\formerqC^2\cdot 2^{-n}\le \formerqC\cdot 2^{-n/2}\le \formerqC\sqrt{\formerqQ}\cdot 2^{-n/2}$ and so
\begin{align*}
	&2\formerqC^2\cdot 2^{-n}+2^{-n/2}\left(8\formerqC\sqrt{\formerqQ}+2\formerqQ\sqrt{2\formerqC}\right)\\
	&\leq 2 \cdot \formerqC\sqrt{\formerqQ}\cdot 2^{-n/2}+2^{-n/2}\left(8\formerqC\sqrt{\formerqQ}+2\formerqQ\sqrt{2\formerqC}\right)\\
	&\leq 10 \cdot 2^{-n/2}\left(\formerqC\sqrt{\formerqQ}+\formerqQ\sqrt{\formerqC}\right)\,,
\end{align*} %\nopagebreak[4]
as claimed.
\qed\end{proof}

%\medskip\noindent{\bf Bounding the distinguishability of the experiments.}
%%%
To complete the proof of \expref{Theorem}{thm:full}, we now establish the two lemmas showing that $\Hyb'_j$ is close to~$\Hyb_{j+1}$ and $\Hyb_j$ is close to $\Hyb'_j$   for $0 \leq j < q_E$. %We do this in the two lemmas that follow. %we used in the proof above. %lemmas, i.e. we show that for $j<\formerqC$, $\Hyb_j$ is close to $\Hyb'_j$ and $\Hyb'_j$ is close to $\Hyb_{j+1}$%, and we show that the swap reprogrammings do not interfere with each other except with small probability

\begin{lemma}\label{lem:step-two-perps}
	For $j=0, \ldots, \formerqC-1$, 
	%if $(k_1,k_2)$ are uniform, or $k_1$ is uniform and $k_2=k_1$, then
	\[\Pr[\A(\Hyb_j') = 1] - \Pr[\A(\Hyb_{j+1})=1]| \leq 2 \cdot q_{P,j+1} \sqrt{2\cdot (j+1)/2^n}\,,\]
	%\jnote{I get a bound of $2q_{P,j+1}\sqrt{2\cdot(j+1)/2^n}$.}
	where $q_{P,j+1}$ is the expected number of queries $\A$ makes to $P$ in the $(j+1)$st stage 
	%between the $(j+1)$st and the $(j+2)$nd classical queries 
	in the ideal world (i.e., in $\Hyb_{q_E}$.) %\chen{Can we use the notation $q_{Q_j}$ instead of $\formerqQ^j$, the latter one is confused with the jth power of $\formerqQ$.}\cm{ok, replaced, but with $q_{Q,j}$, hope that is ok.}
\end{lemma}
\begin{proof}
	%We first prove the bound for the case of uniform $k=(k_1, k_2)$. 
	Recall we can write the oracle sequences defined by $\Hyb'_j$ and $\Hyb_{j+1}$~as
	\begin{alignat*}{5}
	\Hyb'_j: \;\; &P, R, P, \cdots, R, P, ~&&R, \;&&P_{T_{j+1}, k},\;\; &&E_k[P], P_{T_{j+1}, k}, \cdots, E_k[P], P_{T_{j+1}, k}\\
	\Hyb_{j+1}: \;\; &\underbrace{P, R, P, \cdots, R, P}_{\mbox{\scriptsize $j$ classical queries}}, ~&&R, &&P, ~~~&&\underbrace{E_k[P], P_{T_{j+1}, k}, \cdots, E_k[P], P_{T_{j+1}, k}}_{\mbox{\scriptsize $\formerqC-j-1$ classical queries}}\, .
	\end{alignat*}
	%Recall that $\Hyb_j$ and $\Hyb_j'$ are defined by the oracle sequences
	%\begin{alignat*}{5}
	%&R^{\pm}, P^{\pm}; \cdots; R^{\pm}, P^{\pm}; ~~~&&E_k^{\pm}, &&P_{T_j, k}^{\pm}; ~~~&&E_k^{\pm}, P_{T_j, k}^{\pm}~~\,; \cdots; E_k^{\pm}, P_{T_j, k}^{\pm} \\
	%&\underbrace{R^{\pm}, P^{\pm}; \cdots; R^{\pm}, P^{\pm};}_{j} &&R^{\pm}, &&P_{T_{j+1}, k}^{\pm}; &&\underbrace{E_k^{\pm}, P_{T_{j+1}, k}^{\pm}; \cdots; E_k^{\pm}, P_{T_{j+1}, k}^{\pm}}_{q-j-1}\,,
	%\end{alignat*}
	%where $E_k$ is shorthand for $E_k[P]$.
	Let $\A$ be a distinguisher between $\Hyb'_j$ and $\Hyb_{j+1}$. We construct from $\A$ a distinguisher $\D$ for the blinding experiment from \expref{Lemma}{lem:ada-blind}:
	\begin{description}
		\item[Phase 1:] %In the first phase, the task of $\algo D$ is to generate a description of a function $F$ and a blinding algorithm $\mathcal B$, and send them to the challenger $\algo C$ of the blinding experiment. In this phase, $\algo D$ proceeds as follows. 
		$\D$ samples $P, R \leftarrow \Perms_n$. It then runs $\A$, answering its quantum queries using $P$ and its classical queries using~$R$, until after it responds to $\A$'s $(j+1)$st classical query.
		Let $T_{j+1} = \big((x_1, y_1), \dots, (x_{j+1}, y_{j+1})\big)$ be the list of input/output pairs $\A$ received from its classical oracle thus far.
		$\D$ defines $F(t, x) := P^{t}(x)$ for $t \in \{1, -1\}$. It also defines the following randomized algorithm~$\B$: sample $k \leftarrow D$ and then compute the set $B$ of input/output pairs to be reprogrammed so that $F^{(B)}(t, x)=P^t_{T_{j+1}, k}(x)$ for all $t, x$.

		\item[Phase 2:] $\B$ is run to generate~$B$, and
		$\D$ is given quantum access to an oracle~$F_b$. $\D$ resumes running~$\A$, answering its quantum queries using $P^{t} = F_b(t, \cdot)$. Phase~2 ends when $\A$ makes its next (i.e., $(j+2)$nd) classical query.
		
		\item[Phase 3:] $\D$ is given the randomness used by~$\B$ to generate~$k$. It resumes running $\A$, answering its classical queries using $E_k[P]$ and its quantum queries using~$P_{T_{j+1},k}$. Finally, it outputs whatever $\A$ outputs.
	\end{description}
	
	Observe that $\D$ is a valid distinguisher for the reprogramming experiment of \expref{Lemma}{lem:ada-blind}. It is immediate that if $b=0$ (i.e., $\D$'s oracle in phase~2 is~$F_0=F$), then $\A$'s output is identically distributed to its output in~$\Hyb_{j+1}$, whereas if $b=1$ (i.e., $\D$'s oracle in phase~2 is~$F_1=F^{(B)}$), then $\A$'s output is identically distributed to its output in~$\Hyb'_j$. It follows that $|\Pr[\A(\Hyb_j') = 1] - \Pr[\A(\Hyb_{j+1})=1]|$ is equal to the distinguishing advantage of $\D$ in the reprogramming experiment. To bound this quantity using \expref{Lemma}{lem:ada-blind}, we bound the reprogramming probability~$\epsilon$ and the expected number of queries made by $\D$ in phase~2 (when $F = F_0$.)
	
The reprogramming probability $\epsilon$ can be bounded using the definition of $P_{T_{j+1}, k}$ and the fact that $F^{(B)}(t, x) = P^t_{T_{j+1}, k}$. Fixing $P$ and $T_{j+1}$, the probability that any given $(t, x)$ is reprogrammed is at most the probability (over $k$) that it is in the set
	$$
	\left\{(1, x_i \oplus k_1), (1, P^{-1}(y_i \oplus k_2)), (-1, P(x_i \oplus k_1)), (-1, y_i \oplus k_2)\right\}_{i=1}^{j+1}\,.
	$$
Taking a union bound and applying the fact that the marginal distributions of $k_1$ and $k_2$ are each uniform, we get $\epsilon \leq 2 (j+1) / 2^n$.	
	
The expected number of queries made by $\D$ in Phase~2 when $F=F_0$ is equal to the expected number of queries made by $\A$ in its $(j+1)$st stage in~$\Hyb_{j+1}$. 
Since $\Hyb_{j+1}$ and $\Hyb_{q_E}$ are identical until after the $(j+1)$st stage is complete, this is precisely~$q_{P,j+1}$. 
	\qed\end{proof}

\def\bad{{\sf Bad}}
\def\qcoll{{\sf bad}_1}
\def\scoll{{\sf bad}_2}
\def\find{{\sf bad}_3}

\begin{lemma}\label{lem:step-one-perps}
For $j=0, \ldots, \formerqC$, 
%if $(k_1,k_2)$ are uniform, or $k_1$ is uniform and $k_2=k_1$, then
\[ %\label{eq:arl-based-lem-1}
	\left|\Pr[\A(\Hyb_j) = 1] - \Pr[\A(\Hyb'_j)=1]\right| \leq 8 \cdot \sqrt{\frac{\formerqQ}{2^n}}+2q_E\cdot2^{-n}.
\]
%
%If instead $k_1$ is uniform and $k_2=k_1$ then
%\begin{equation}\label{eq:arl-based-lem-2}
%	\left|\Pr[\A(\Hyb_j) = 1] - \Pr[\A(\Hyb'_j)=1]\right| \leq 8 \cdot \sqrt{\frac{\formerqQ}{2^n}}+(5+2n)\cdot q_E\cdot2^{-n}.
%	\end{equation}
%\jnote{I get a bound of $8 \sqrt{q_P/2^n} + 3 q_E/2^n$.}
\end{lemma}
\begin{proof}
Recall that we can write the oracle sequences defined by $\Hyb_j$ and $\Hyb_j'$~as
	\begin{alignat*}{5}
	\Hyb_j: \;\; &P, R, P, \cdots, R, P, ~~&&E_k[P],\; &&P_{T_j, k}, ~~~&&E_k[P], P_{T_j, k}~~\,, \cdots, E_k[P], P_{T_j, k} \\
	\Hyb'_j: \;\; &\underbrace{P, R, P, \cdots, R, P}_{\mbox{\scriptsize $j$ classical queries}}, ~~&&R, &&P_{T_{j+1}, k}, &&\underbrace{E_k[P], P_{T_{j+1}, k}, \cdots, E_k[P], P_{T_{j+1}, k}}_{\mbox{\scriptsize $\formerqC-j-1$ classical queries}}\,.
	\end{alignat*}
%\jnote{I dropped the $(\cdot)^\pm$ notation.}
%	
%Recall that $\Hyb_j$ and $\Hyb_j'$ are defined by the oracle sequences
%\begin{alignat*}{5}
%&R^{\pm}, P^{\pm}; \cdots; R^{\pm}, P^{\pm}; ~~~&&E_k^{\pm}, &&P_{T_j, k}^{\pm}; ~~~&&E_k^{\pm}, P_{T_j, k}^{\pm}~~\,; \cdots; E_k^{\pm}, P_{T_j, k}^{\pm} \\
%&\underbrace{R^{\pm}, P^{\pm}; \cdots; R^{\pm}, P^{\pm};}_{j} &&R^{\pm}, &&P_{T_{j+1}, k}^{\pm}; &&\underbrace{E_k^{\pm}, P_{T_{j+1}, k}^{\pm}; \cdots; E_k^{\pm}, P_{T_{j+1}, k}^{\pm}}_{q-j-1}\,,
%\end{alignat*}
%where $E_k$ is shorthand for $E_k[P]$.
Let $\A$ be a distinguisher between $\Hyb_j$ and $\Hyb_j'$. 
We construct from $\A$ a distinguisher $\D$ for the reprogramming experiment of \expref{Lemma}{lem:prp-arl}:
\begin{description}
\item[Phase 1:]  $\D$ is given quantum access to a permutation~$P$. It samples \mbox{$R \from \Perms_n$} and then runs $\A$, answering its quantum queries with $P$ and its classical queries with $R$ (in the appropriate directions), until\footnote{We assume for simplicity that this query is in the forward direction, but the case where it is in the inverse direction can be handled entirely symmetrically (using the fact that the marginal distribution of $k_2$ is uniform). The strings $s_0$ and $s_1$ are in that case replaced by $P_b(s_0)$ and~$P_b(s_1)$. See Appendix~\ref{appendix:inverse-case} for details.} $\A$ submits its \mbox{$(j+1)$st} classical query~$x_{j+1}$. At that point, $\D$ has a list $T_j=\big((x_1, y_1), \cdots, (x_j, y_j)\big)$ of the input/output pairs $\A$ has received from its classical oracle thus far. 
\item[Phase 2:]  
%\cm{I think it would be good if we described what the distinguisher does if $\A$'s  \mbox{$(j+1)$st} classical query is in the inverse direction. The reason is that the distinguisher needs to make 2 extra queries here to determine $P_b(s_i)$. $P_b(s_0)$ (say) is used directly to compute $k_2$, and $P_b(s_1)$ is needed for the reduction to bound $\Pr[\find]$} \jnote{Hmm, good point, I need to think about this.}
%\chen{I tried to simulate the case when the (j+1)st query is $y_{j+1}$ and it looks symmetric to me (please check this makes sense or not): $\D$ receives uniform $t_0, t_1 \in \bool^n$, where $t_0=P(s_0)$ and $t_1=P(s_1)$ for some uniform $s_0, s_1 \in \bool^n$, then $\D$ sets $k_2=t_1 \oplus y_{j+1}=P(s_1)\oplus y_{j+1}$, choose  $k_1 \leftarrow D_{|k_2}$, here when $b=1$ we could set 
%$$P_1^{-1} = P^{-1} \circ \swap_{t_0, t_1} = (\swap_{t_0, t_1} \circ P)^{-1} =(\swap_{P(s_0), P(s_1)} \circ P)^{-1}= (P \circ \swap_{s_0,s_1})^{-1}$$
%and then
%$$x_{j+1} \stackrel{\rm def}{=} E_k^{-1}[P_1](y_{j+1}) = P_1^{-1}(y_{j+1} \oplus k_2) \oplus k_1 = P_1^{-1}(P(s_1)) \xor k_1 = s_0 \oplus k_1\,.$$
%The rest are same as the proof.} \ga{I agree that the footnote is not quite sufficient. (In particular, it's not clear what is meant by ``symmetrically.'') I'll come back to this after I finish my pass on the lemma.}\cm{For now, for the arXiv posting, I added a sentence to the footnote.}

Now $\D$ receives $s_0, s_1 \in \bool^n$ and quantum oracle access to a permutation~$P_b$. Then $\D$ sets $k_1:=s_0 \xor x_{j+1}$, chooses $k_2 \leftarrow D_{|k_1}$ (where this represents the conditional distribution on $k_2$ given~$k_1$), and
%\begin{itemize}
%\item If $\A$'s $(j+1)$st classical query (call it $x^*$) was in the forward direction, then $\D$ sets $k_1:=s_0 \xor x^*$ and chooses uniform~$k_2$.
%\item If $\A$'s $(j+1)$st classical query (call it $y^*$) was in the inverse direction, then $\D$ sets $k_2:=s_1 \xor y^*$ \jnote{This doesn't seem right to me. Shouldn't it be $k_2:=P(s_1) \xor y^*$?}\cm{oh. This is a problem, isn't it? Because the distinguisher in the resampling game only gets $s_0, s_1$ and has no way of obtaining $P(s_1)$... I think it would be better to do what I had originally in mind: Just construct a distinguisher for the resampling game of $P^{-1}$ if the $(j+1)$st classical query is to $P^{-1}$} \jnote{Perhaps it should be $k_2:=P_b(s_1) \xor y^*$ instead}\chen{To correspond with later settings, perhaps it's $k_2:=P_b(s_0) \xor y^*$. } \cm{I think you're right, Chen, setting $k_2:=P_b(s_0) \xor y^*$ should work. For $b=0$ that value is uniform, as $s_0$ is used nowhere else, and for $b=1$ it does the right thing.} and chooses uniform~$k_1$. 
%\leftarrow D(\cdot|k_2)
%\end{itemize}
%(We write $D(k_1|\cdot)$ for the conditional distribution of $k_2$ given~$k_1$, and similarly for $D(\cdot|k_2)$.)
%In either case, 
sets $k:=(k_1, k_2)$. $\D$ continues running $\A$, answering its remaining classical queries (including the $(j+1)$st one) using $E_k[P_b]$, and its remaining quantum queries using 
\[
%[P_b]_{T, k}^\pm = 
(P_b)_{T_j, k} = \swap_{P_b(x_1 \oplus k_1), y_1 \oplus k_2} \circ \cdots \circ \swap_{P_b(x_j \oplus k_1), y_j \oplus k_2} \circ P_b\,.
%P_b\circ\swap_{(x_j \oplus k_1,P^{-1}( y_j \oplus k_2))} \circ \cdots \circ \swap_{(x_1 \oplus k_1,P^{-1}( y_1 \oplus k_2))}\,.
\]
Finally, $\D$ outputs whatever $\A$ outputs.
\end{description}
Note that although $\D$ makes additional queries to $P_b$ at the start of phase~2 (to determine $P_b(x_1 \xor k_1), \ldots, P_b(x_j \xor k_1)$), %This is necessary for constructing the phase~2 quantum oracle for $\A$. 
the bound of \expref{Lemma}{lem:prp-arl} only depends on the number of quantum queries $\D$ makes in phase~1, which is at most~$q_P$. %, which is the same as the number of quantum queries made by $\A$ in phase~1.

We now analyze the execution of $\algo D$ in the two cases of the game of \expref{Lemma}{lem:prp-arl}: $b=0$ (no reprogramming) and $b=1$ (reprogramming). 
%We assume for simplicity that $\A$'s $(j+1)$st classical query is in the forward direction, but one can verify that everything holds analogously if the query is in the inverse direction. 
In both cases, $P$ and~$R$ are independent, uniform permutations, and $\A$ is run with quantum oracle $P$ and classical oracle~$R$ until it makes its $(j+1)$st classical query; thus,
through the end of phase~1, the above execution of $\A$ is consistent with both $\Hyb_j$ and~$\Hyb_j'$.

At the start of phase~2, uniform $s_0, s_1 \in \bool^n$ are chosen. Since $\algo D$ sets $k_1 := s_0 \oplus x_{j+1}$,
% and samples an independent, uniform~$k_2$, 
%\leftarrow D(k_1|\cdot)
the distribution of $k_1$ is uniform and hence~$k$ is distributed according to~$D$.
%exactly that given by~$D$ (using the fact that $D$ has uniform marginals)
%Let $T_j=\big((x_1, R(x_1)), \cdots, (x_j, R(x_j))\big)$ be the input/output pairs for the first $j$ classical queries of~$\A$. 
The two cases ($b=0$ and $b=1$) now begin to diverge.

\medskip\noindent \textbf{Case $b=0$ (no reprogramming).} 
In this case, $\A$'s remaining classical queries (including its $(j+1)$st classical query) are answered using $E_k[P_0] = E_k[P]$, and its remaining quantum queries are answered using $(P_0)_{T_j, k} = P_{T_j, k}$. The output of $\A$ is thus distributed identically to its output in~$\Hyb_j$ in this case.

\medskip\noindent \textbf{Case $b=1$ (reprogramming).} 
In this case, we have
\begin{eqnarray}\label{eq:def-of-P1}
P_b = P_1 = P \circ \swap_{s_0, s_1} = \swap_{P(s_0), P(s_1)} \circ P = \swap_{P(x_{j+1} \oplus k_1), P(s_1)} \circ P\,.
\end{eqnarray}
The response to $\A$'s $(j+1)$st classical query  is thus
\begin{eqnarray}\label{eq:repro=lazy}
y_{j+1} \stackrel{\rm def}{=} E_k[P_1](x_{j+1}) = P_1(x_{j+1} \oplus k_1) \oplus k_2 = P_1(s_0) \xor k_2 = P(s_1) \oplus k_2\,.
\end{eqnarray}
The remaining classical queries of $\A$ are then answered using $E_k[P_1]$, while its remaining quantum queries are answered using~$(P_1)_{T_j, k}$. 
If we let ${\sf Expt}_j$ refer to the experiment in which $\D$ executes $\A$ as a subroutine when $b=1$, it follows from \expref{Lemma}{lem:prp-arl} that
\begin{eqnarray}
	\left|\Pr[\A(\Hyb_j) = 1] - \Pr[\A({\sf Expt}_j)=1]\right| \leq 4  \sqrt{\formerqQ/2^n}. \label{eqn:bound:H-E}
\end{eqnarray}

We now define three events: 
\begin{enumerate}
		\item $\qcoll$ is the event that $y_{j+1} \in \{y_1, \ldots, y_j\}$.
\item $\scoll$  is the event that $s_1 \xor k_1 \in \{x_1, \ldots, x_j\}$. %\jnote{Note that $x_i = s_0 \xor k_1$ is disallowed by assumption.}
\item $\find$  is the event that, in phase~2, $\A$ queries its classical oracle in the forward direction on $s_1 \xor k_1$, %=s_0\oplus s_1\oplus x^*$, 
or the inverse direction on $P(s_0) \xor k_2$ (with result $s_1\oplus k_1$). 
\end{enumerate}
%We begin by upper-bounding the probabilities of each of these events. 
Since $y_{j+1}=P(s_1) \xor k_2$ is uniform (because $k_2$ is uniform and independent of $P$ and~$s_1$), it is immediate that $\Pr[\qcoll] \leq j/2^n$. Similarly, $s_1 \xor k_1 = s_1 \xor s_0 \xor x_{j+1}$  is uniform, and so
$\Pr[\scoll] \leq j/2^n$. 
As for the last event, 
we have:
\begin{claim}\label{find-claim}
$\Pr[\find] \le (q_E-j)/2^{n}+4 \sqrt{q_P/2^n}$. %\le 2(q_E-j)/2^{n}+4 \sqrt{q_P/2^n}$.
\end{claim}
%\chen{More precisely I think $\Pr[\find] \le (q_E-j)/2^{n}+4 \sqrt{q_P/2^n}$.so that the upper bound of equation (4) could be $2q_E/2^n + 4\sqrt{q_P/2^n}$, though I don't think that really matters.} 
\begin{proof}
	%We again use \expref{Lemma}{lem:prp-arl}.
	Consider the algorithm $\D'$ that behaves identically to~$\D$ in phases~1 and~2, but then when $\A$ terminates outputs~1 iff event~$\find$ occurred. When $b=0$ (no reprogramming), the execution of $\A$ is independent of $s_1$, and so the probability that $\find$ occurs is at most~$(q_E-j)/2^n$. Now observe that $\D'$ is a distinguisher for the reprogramming game of \expref{Lemma}{lem:prp-arl}, with advantage $|\Pr[\find|b=1] - (q_E-j)/2^n|$. The claim then follows from \expref{Lemma}{lem:prp-arl}.
	%but terminates and outputs $1$ if $\mathsf{Find}$ occurs and outputs $0$ otherwise.
	\qed\end{proof}

	\begin{figure}[h]
	%	\begin{center}
	\framebox[\textwidth][l]{
		%	\begin{minipage}{\textwidth}
		\begin{algorithm}[H]
			\DontPrintSemicolon
			$P, R \from \Perms_n$\;
			Run $\A$ with quantum access to $P$ and classical access to $R$,   until $\A$ makes its $(j+1)$st classical query~$x_{j+1}$;  let $T_j$ %:=\big((x_1,y_1), \cdots, (x_j,y_j)\big)$ is 
			be as in the text\;
			$s_0, s_1 \leftarrow \bool^n$,  $\;P_1 := P \circ {\sf swap}_{s_0, s_1}$\;
			$k_1 := s_0 \xor x_{j+1}$,  $k_2 \leftarrow D_{\mid k_1}$, $k:=(k_1, k_2)$ \;
			$y_{j+1} := E_k[P_1](x_{j+1})$ \; %=P(s_1) \xor k_2
			$Q:=(P_1)_{T_j, k}$\;
			\lIf{$y_{j+1} \in \{y_1, \ldots, y_j\}$}{
				$\qcoll:={\sf true}$, \fbox{$y_{j+1} \leftarrow \bool^n\setminus\{y_1, \ldots, y_j\}$} \label{alg:qcoll}}
			Give $y_{j+1}$ to $\A$ as the answer to its $(j+1)$st classical query\;
			$T_{j+1} := \big((x_1, y_1), \ldots, (x_{j+1}, y_{j+1})\big)$\;
			\lIf{$s_1 \xor k_1 \in \{x_1, \ldots, x_j\}$}{
				$\scoll:={\sf true}$}
			\lIf{$\qcoll={\sf true}$ or $\scoll={\sf true}$}{
				\fbox{$Q:=P_{T_{j+1},k}$} \label{alg:setQ}}
			Continue running $\A$ with quantum access to $Q$ and classical access to~$\O/\O^{-1}$  \label{alg:oracles} \;
		\end{algorithm}
		%\end{minipage}
	}
	\begin{tabular}[t]{ll}
		\framebox[0.49\textwidth][l]{
			\begin{algorithm}[H]
				\setcounter{AlgoLine}{12}
				\DontPrintSemicolon
				\underline{$\O(x)$}   \;
				$y:=E_k[P_1](x)$ \label{alg:compute-y} \;
				\If{$x=s_1 \xor k_1$}{
					$\find:={\sf true}$, \fbox{$y:=E_k[P](x)$}}
				\KwRet{$y$\vspace*{1pt}}
		\end{algorithm}} & 
		\framebox[0.49\textwidth][l]{
			\begin{algorithm}[H]
				\setcounter{AlgoLine}{17}
				\DontPrintSemicolon
				\underline{$\O^{-1}(y)$}\;
				$x:=E^{-1}_k[P_1](y)$\;
				\If{$x=s_1 \xor k_1$}{
					$\find:={\sf true}$, \fbox{$x:=E^{-1}_k[P](y)$}}
				\KwRet{$x$}
		\end{algorithm}}
	\end{tabular}
	\caption{${\sf Expt}'_j$ includes the boxed statements, whereas ${\sf Expt}_j$ does not. \label{fig:two-games}} %\vspace*{-.5in}
	%	\end{center}
\end{figure}

In Figure~\ref{fig:two-games}, we show code for ${\sf Expt}_j$ and a related experiment~${\sf Expt}'_j$. 
Note that ${\sf Expt}_j$ and ${\sf Expt}'_j$ are identical until either $\qcoll, \scoll$, or $\find$ occur, and so by the fundamental lemma of game playing\footnote{This lemma is an information-theoretic result, and can be applied in our setting since everything we say in what follows holds even if $\A$ is given the entire function table for its quantum oracle~$Q$ in line~\ref{alg:oracles}. }~\cite{DBLP:conf/eurocrypt/BellareR06} we have
\begin{eqnarray}
\left|\Pr[\A({\sf Expt}'_j) = 1] - \Pr[\A({\sf Expt}_j)=1]\right| & \leq & \Pr[\qcoll \vee \scoll \vee \find] \nonumber \\
& \leq & 2q_E/2^n + 4\sqrt{q_P/2^n}\, . \label{eqn:bound:E-F}
\end{eqnarray}
We complete the proof by arguing that ${\sf Expt}'_j$ is identical to $\Hyb'_j$:
\begin{enumerate}
	    \item In ${\sf Expt}'_j$, the oracle $Q$ used in line~\ref{alg:oracles} is always equal to $P_{T_{j+1}, k}$. When $\qcoll$ or $\scoll$ occurs this is immediate (since then $Q$ is set to $P_{T_{j+1}, k}$ in line~\ref{alg:setQ}).  But if $\qcoll$ does not occur then Equation~(\ref{eq:repro=lazy}) holds, and if
	    $\scoll$ does not occur then for $i=1, \ldots, j$ we have $x_i \xor k_1 \neq s_0$ and
	$x_i \xor k_1 \neq s_1$ (where the former is because $x_{j+1} \xor k_1 = s_0$ but $x_i \neq x_{j+1}$ by assumption, and the latter is by definition of $\scoll$). Thus $P_1(x_i \xor k_1) = P(x_i \xor k_1)$ for  $i=1, \ldots, j$, and so
	\begin{align*}
	Q=(P_1)_{T_j, k} &= \swap_{P_1(x_1 \oplus k_1), y_1\oplus k_2} \circ \cdots \circ \swap_{P_1(x_j \oplus k_1), y_j \oplus k_2} \circ P_1\\
	& = \swap_{P(x_1 \oplus k_1), y_1\oplus k_2} \circ \cdots  \circ\swap_{P(x_{j+1} \oplus k_1), y_{j+1} \oplus k_2} \circ P \\
	& = P_{T_{j+1}, k},
	\end{align*}
	using Equations~(\ref{eq:def-of-P1}) and~(\ref{eq:repro=lazy}). 
	
	\item In ${\sf Expt}'_j$, the value $y_{j+1}$ is uniformly distributed in $\bool^n \setminus \{y_1, \ldots, y_j\}$. Indeed, we have already argued above that the value $y_{j+1}$ computed in line~\ref{alg:compute-y} is uniform in~$\bool^n$. But if that value lies in $\{y_1, \ldots, y_j\}$ (and so $\qcoll$ occurs) then $y_{j+1}$ is re-sampled uniformly from $\bool^n \setminus \{y_1, \ldots, y_j\}$ in line~\ref{alg:qcoll}.
	
    \item In ${\sf Expt}'_j$, the response from oracle $\O(x)$ is always equal to $E_k[P](x)$. When $\find$ occurs this is immediate. But if $\find$ does not occur then $x \neq s_1 \xor k_1$; we also know that $x \neq s_0 \xor k_1=x_{j+1}$ by assumption. But then 
    $P_1(x \xor k_1) = P(x \xor k_1)$ and so $E_k[P_1](x) = E_k[P](x)$. A similar argument shows that the response from $\O^{-1}(y)$ is always
    $E_k^{-1}[P](y)$.
\end{enumerate}
Syntactically rewriting ${\sf Expt}'_j$ using the above observations yields an experiment that is identical to~$\Hyb'_j$. (See Appendix~\ref{appendix:games} for further details.)
\expref{Lemma}{lem:step-one-perps} thus follows from Equations~(\ref{eqn:bound:H-E}) and~(\ref{eqn:bound:E-F}).
\qed\end{proof}

\section{Proofs of the Technical Lemmas}\label{sec:lemma-proofs}
%%%%%%%%%%%%%%%%%%%%%%%%%%%

In this section, we give the proofs of our technical lemmas: the ``arbitrary reprogramming lemma'' (\expref{Lemma}{lem:ada-blind}) and the ``resampling lemma'' (\expref{Lemma}{lem:prp-arl}).

\subsection{Proof of the Arbitrary Reprogramming Lemma}\label{sec:reprog-lemma}
%%%%%%%%%%%%%%%%%%%%%%%%%%%

\ignore{
We now prove \expref{Lemma}{lem:ada-blind}. %, restated below for convenience.

\medskip\noindent{\bf Restatement of \expref{Lemma}{lem:ada-blind}.} %\label{lem:ada-blind-proved}
Let $\D$ be a distinguisher in the following game:
\begin{description}
	\item[Phase 1:]  $\D$ outputs descriptions of a function $F_0=F: \bool^m \rightarrow \bool^n$ and a randomized algorithm $\mathcal B$ whose output is a set $B \subset \bool^m \times \bool^n$ where each $x \in \bool^m$ is the first element of at most one tuple in~$B$. Let 
	$
	\epsilon=\max_{x \in \bool^m}\left\{\Pr_{B \leftarrow \B}[x \in B_1]\right\}$.
	
	\item[Phase 2:]  $\mathcal{B}$ is run to obtain~$B$. Let $F_1=F_B$.
	A uniform bit~$b$ is chosen, and
	$\D$ is given quantum access to~$F_b$. 
	\item[Phase 3:]  $\D$ loses access to $F_b$, and  receives the randomness~$r$  used to invoke~$\algo B$ in phase~2. Then $\D$ outputs a guess~$b'$. 
\end{description}
	For any $\D$ making $q$ queries in expectation when its oracle is $F_0$,
it holds that $\left|\Pr[\mbox{$\D$ outputs 1} \mid b=1] - \Pr[\mbox{$\D$ outputs 1} \mid b=0]\right| \leq 2q \cdot \sqrt{\epsilon}$.
}

%An algorithm like the distinguisher $\mathcal D$ described in the last paragraph is not unitary. Because we would like to prove \emph{the same bound} as given in \expref{Lemma}{lem:blind}, 
\expref{Lemma}{lem:ada-blind} allows for distinguishers that choose the number of queries they make adaptively, e.g., depending on the oracle provided and the outcomes of any measurements, and the bound is in terms of the number of queries $\D$ makes \emph{in expectation}. As discussed in \expref{Section}{sec:technical-intro}, the ability to directly handle such adaptive distinguishers is necessary for our proof, and to our knowledge has not been addressed before. 
%Standard proofs of statements like \expref{Lemma}{lem:ada-blind-proved} assume a fixed number of queries, and postpone all measurements to the end. In our case, we must allow measurements throughout, so that $\D$ can decide whether to query based on their outcomes. 
To formally reason about adaptive distinguishers, we model the intermediate operations of the distinguisher and the measurements it makes as \emph{quantum channels}. With this as our goal, we first recall some necessary background and establish some notation.
%A unitary version of the measurement would give rise to terms proportional to $\sqrt{p_i(1-p_i)}$ where $p_i$ is the probability that $\mathcal D$ makes $i$ or more queries, spoiling the bound proportional to the expected number of queries $\sum_{i=1}^{\infty}p_i$. 
%In preparation, we first provide a brief recap of the formalism of density matrices and quantum channels.

Recall that 
a density matrix $\rho$ is a positive semidefinite matrix with unit trace. 
%The density matrix of a pure state $\ket\psi$ is $\proj\psi$. 
%the orthogonal projector onto the line spanned by $\ket\psi$, which is denoted by 
A quantum channel---the most general transformation between density matrices allowed by quantum theory---is a completely positive, trace-preserving, linear map. 
%: \C^{d_1\times d_1}\to  \C^{d_2\times d_2}$. 
The quantum channel corresponding to the unitary operation $U$ is the map 
$\rho \mapsto U\rho U^\dagger$.
Another type of quantum channel is a \emph{pinching}, which corresponds to the operation of making a measurement. Specializing to the only kind of pinching needed in our proof, consider the measurement of a single-qubit register~$C$ given by the projectors $\{\Pi_0, \Pi_1\}$ with $\Pi_b = \ket{b}\bra{b}_C$. This corresponds to the pinching~${\mathcal M}_C$ where
\[{\mathcal M}_C (\rho) = \Pi_0 \rho \Pi_0 + \Pi_1 \rho \Pi_1.\]
Observe that a pinching only produces the post-measurement state, and does not separately give the outcome (i.e., the result $0$ or $1$).

%Consider a projective measurement given by a set of projectors $\{\Pi_i\}_{i=1}^\ell$ where $\sum_{i=1}^\ell\Pi_i=\mathds 1$. \jnote{Since it is all we need, maybe we can specialize to the case of measuring a single-qubit register? I also notice that we never use the notation ${\mathcal P}(\cdot)$ again; instead, we make up new notation ${\mathcal M}_C$. Why not define that notation instead?} The corresponding pinching $\mathcal P$ is defined by $\mathcal P(X)=\sum_{i=1}^\ell \Pi_iX\Pi_i$. In other words, the pinching is the quantum channel corresponding to the operation that applies a measurement, discards the result, and retains the post-measurement state.

Consider a quantum algorithm $\mathcal D$ with access to an oracle~$\O$ operating on registers $X, Y$ (so $\O \ket{x}\ket{y} = \ket{x}\ket{y \oplus \O(x)}$). We define the unitary $c\O$ for the \emph{controlled} version of~$\O$, operating on registers $C, X$, and $Y$ (with $C$ a single-qubit register), as 
%\begin{equation*}
%c\O_{CXY}=\O_{XY}\otimes \proj 0_C+\mathds 1_{XY}\otimes \proj 1_C.
%\end{equation*}
\[c\O\ket{c}\ket{x}\ket{y} = \ket{c}\ket{x}\ket{y\oplus c\cdot \O(x)}.\]
%\jnote{Is not it more typical to have $\O$ be applied when $c=1$?}
%\jnote{Could we more simply write $c\O_{XYC} \ket{x}\ket{y}\ket{c} = \ket{x}\ket{y\oplus c\cdot \O(x)}\ket{c}$?}
With this in place, 
we may now view an execution of $\mathcal D^O$ as follows. The algorithm uses registers $C, X, Y$, and~$E$. 
Let $q_{\max}$ be an upper bound on the number of queries $\mathcal D$ ever makes. 
Then $\mathcal D$ applies the quantum channel 
\begin{equation}\label{eq:adaptive-distinguisher}
\left(\Phi \circ c\mathcal O\circ\mathcal M_C\right)^{q_{\max}}
\end{equation}
to some initial state~$\rho=\rho_0^{(0)}$. 
That is, for each of $q_{\max}$ iterations, $\mathcal{D}$ applies to its current state the pinching $\mathcal{M}_C$ followed by the controlled oracle 
$c\mathcal O$ and then an arbitrary
quantum channel 
$\Phi$  (that we take to be the same in all iterations without loss of generality\footnote{This can be done by having a register serve as a counter that is incremented with each application of $\Phi$.}) operating on all its registers.
Finally, $\mathcal D$ applies a measurement to produce its final output. %We omit this measurement here, bearing in mind that (ignoring time complexity) we can always take it to be the Helstrom measurement for the resulting state discrimination problem. 
If we let $\rho^{(0)}_{i-1}$ denote the intermediate state immediately before the pinching is applied in the $i$th iteration, then 
$p_{i-1}=\Tr\left[\proj 1_C\, \rho^{(0)}_{i-1}\right]$ represents the probability that the oracle is applied (or, equivalently, that a query is made) in the $i$th iteration, and so 
$q = \sum_{i=1}^{q_{\max}} p_{i-1}$ is the expected number of queries made by~$\D$ when interacting with oracle~$\O$.

\medskip \noindent {\bf Proof of \expref{Lemma}{lem:ada-blind}.}
%%%
An execution of $\D$ takes the form of Equation~(\ref{eq:adaptive-distinguisher}) up to a final measurement. 
For some fixed value of the randomness $r$ used to run~$\B$, set $\Upsilon_b=\Phi \circ c\mathcal O_{F_b}\circ\mathcal M_C$, and define
	\[\rho_k \stackrel{{\rm def}}{=} \left(\Upsilon_1^{q_{\max}-k} \circ \Upsilon_0^k\right) (\rho),\]
so that $\rho_k$ is the final state if the first $k$ queries are answered using a (controlled) $F_0$ oracle and then the remaining $q_{\max}-k$ queries are answered using a (controlled) $F_1$ oracle.	
Furthermore, we define $\rho_i^{(0)} = \Upsilon_0^i(\rho)$.
Note also that $\rho_{q_{\max}}$ (resp., $\rho_0$) is the final state of the algorithm when the $F_0$ oracle (resp., $F_1$ oracle) is used the entire time.
We bound ${\mathbb E}_r\left[\delta\left(\proj r\otimes \rho_{q_{\max}},\rule{0pt}{9pt} \, \proj r \otimes \rho_0\right)\right]$, where $\delta(\cdot, \cdot)$ denotes the trace distance. 
%\chen{One reviewer mentioned the trace distance defined here is inconsistent with the one in section 4.2 (which is $||\cdot-\cdot ||_1$), do we need to change it? } \jnote{We are not inconsistent. We don't directly use trace distance in Section~4.2; we just mention it for an equality we use.}
	
Define $\tilde F^{(B)}(x)=F(x) \xor F^{(B)}(x)$, and note that $\tilde F^{(B)}(x)=0^n$ for $x \not \in B_1$. 
Since trace distance is non-increasing under quantum channels, for any $r$ we have
\begin{eqnarray*}
\delta\left(\proj r\otimes \rho_{k}, \, \proj r\otimes \rho_{k-1}\right) &\le &\delta\left({c\mathcal O}_{F_0}\circ\mathcal M_C\left(\rho_{k-1}^{(0)}\right),  \; {c\mathcal O}_{F_1}\circ\mathcal M_C\left(\rho_{k-1}^{(0)}\right)\right)\\
&=&\delta\left(\mathcal M_C\left(\rho_{k-1}^{(0)}\right),  \; {c\mathcal O}_{\tilde F^{(B)}}\circ\mathcal M_C\left(\rho_{k-1}^{(0)}\right)\right).
\end{eqnarray*}
By definition of a controlled oracle,
\begin{eqnarray*}
{c\mathcal O}_{\tilde F^{(B)}}\circ\mathcal M_C\left(\rho_{k-1}^{(0)}\right) &=& {c\mathcal O}_{\tilde F^{(B)}}\left(\proj 1_C\,\rho_{k-1}^{(0)}\,\proj 1_C\right)+\proj 0_C\,\rho_{k-1}^{(0)}\,\proj 0_C
\\
&=& {\mathcal O}_{\tilde F^{(B)}}\left(\proj 1_C\,\rho_{k-1}^{(0)}\,\proj 1_C\right)+\proj 0_C\,\rho_{k-1}^{(0)}\,\proj 0_C,
\end{eqnarray*}
and thus
\begin{eqnarray*}
\lefteqn{\delta\left(\mathcal M_C\left(\rho_{k-1}^{(0)}\right),  \; {c\mathcal O}_{\tilde F^{(B)}}\circ\mathcal M_C\left(\rho_{k-1}^{(0)}\right)\right)} \\
&=&\delta\left(\proj 1_C\,\rho_{k-1}^{(0)}\,\proj 1_C, \; {\mathcal O}_{\tilde F^{(B)}}\left(\proj 1_C\,\rho_{k-1}^{(0)}\,\proj 1_C\right)\right)\\
&=&p_{k-1} \cdot \delta\left(\sigma_{k-1},  \; {\mathcal O}_{\tilde F^{(B)}}\left(\sigma_{k-1}\right)\right)
\end{eqnarray*}
where, recall, $p_{k-1}=\Tr\left[\proj 1_C\,\rho_{k-1}^{(0)}\right]$ is the probability that a query is made in the $k$th iteration, and we define the normalized state 
$
\sigma_{k-1}\stackrel{\rm def}{=}\frac{\proj 1_C\, \rho_{k-1}^{(0)}\,\proj 1_C}{p_{k-1}}$.
Therefore,
\begin{eqnarray}
\lefteqn{\mathbb E_r\left[\delta\left(\proj r\otimes \rho_{q_{\max}}, \; \proj r\otimes \rho_{0}\right)\right]} \nonumber \\
& \leq & \sum_{k=1}^{q_{\max}} {\mathbb E}_B\left[\delta(\left(\proj r\otimes \rho_{k}, \; \proj r\otimes \rho_{k-1}\right)\right] \nonumber \\
&\le&\sum_{k=1}^{q_{\max}}p_{k-1}\cdot \mathbb E_B\left[\delta\left(\sigma_{k-1},  \; {\mathcal O}_{\tilde F^{(B)}}\left(\sigma_{k-1}\right)\right)\right] \nonumber \\
& \leq & q \cdot \max_{\sigma} \, \mathbb E_B\left[\delta\left(\sigma,  \; {\mathcal O}_{\tilde F^{(B)}}\left(\sigma\right)\right)\right], \label{eqn:adaptive1}
%&\le& 2\sqrt \epsilon  \sum_{k=1}^{q_{\max}}p_{k-1}=2q\sqrt \epsilon,
\end{eqnarray}
where we write $\Exp_B$ for the expectation over the set $B$ output by $\mathcal{B}$ in place of~$\Exp_r$.

Since $\sigma$ can be purified to some state $\ket{\psi}$, and $\delta(\ket{\psi}, \ket{\psi'}) \leq \|\ket{\psi}-\ket{\psi'}\|_2$ for pure states $\ket{\psi}, \ket{\psi'}$, 
we have 
\begin{eqnarray*}
\max_{\sigma} \, \mathbb E_B\left[\delta\left(\sigma,  \; {\mathcal O}_{\tilde F^{(B)}}\left(\sigma\right)\right)\right] 
&\leq & \max_{\ket{\psi}} \, {\mathbb E}_B\left[\delta\left(\ket{\psi},  \; {\mathcal O}_{\tilde F^{(B)}}\ket{\psi}\right)\right] \\
&\leq & \max_{\ket{\psi}} \, {\textstyle \Exp_B \left[\| \ket{\psi} - \O_{\tilde F^{(B)}}\ket{\psi}\|_2\right]}.
\end{eqnarray*} 
%\chen{I'm wondering shall we add some steps to prove the first inequality of the above equation is true. For example, suppose $\sigma \in \mathcal{H}_X$, then there exists a pure state $\ket{\psi}\in \mathcal{H}_X \otimes \mathcal{H}_Y$ such that 
%$$
%\sigma = \Tr_Y(\proj{\psi} )
%$$
%Then we have 
%\begin{eqnarray*}
%\max_{\sigma} \, \mathbb E_B\left[\delta\left(\sigma,  \; {\mathcal O}_{\tilde F^{(B)}}\left(\sigma\right)\right)\right]
%&=  &\max_{\sigma} \, \mathbb E_B\left[\delta\left( \Tr_Y\left(\proj{\psi} \right),  \; {\mathcal O}_{\tilde F^{(B)}}\left( \Tr_Y\left(\proj{\psi}
%\right)\right)\right)\right] \\
%&=  &\max_{\sigma} \, \mathbb E_B\left[\delta\left( \Tr_Y\left(\proj{\psi}\right ),  \;  \Tr_Y\left ({\mathcal O}_{\tilde F^{(B)}}\left(\proj{\psi}
%\right)\right)\right)\right] \\
%&\leq &\max_{\sigma} \, \mathbb E_B\left[\delta\left( \proj{\psi} ,  \;  {\mathcal O}_{\tilde F^{(B)}}\left(\proj{\psi}
%\right)\right)\right]
%\end{eqnarray*}
%since $\frac{1}{2}\Tr(\Tr_Y(\rho)-\Tr_Y(\epsilon))\leq\frac{1}{2}\Tr(\rho-\epsilon)$. Also should we replace $\ket{\psi}$ with 
%$\proj{\psi}$?
%}
Because $\O_{\tilde F^{(B)}}$ acts as the identity on $(\mathbb{I}-\Pi_{B_1})\ket{\psi}$ for any $\ket{\psi}$, we have 
\begin{eqnarray}
\lefteqn{ {\textstyle \Exp_B \left[\| \ket{\psi} - \O_{\tilde F^{(B)}}\ket{\psi}\|_2\right]} } \nonumber \\
& = & 
{\textstyle \Exp_B \left[\| \Pi_{B_1} \ket{\psi} - \O_{\tilde F^{(B)}}\Pi_{B_1} \ket{\psi} + (\mathbb{I}-\O_{\tilde F^{(B)}})(\mathbb{I}-\Pi_{B_1} )\ket{\psi}\|_2\right]} \nonumber \\
& \leq & {\textstyle \Exp_B \left[\| \Pi_{B_1} \ket{\psi}\|_2\right]} +  {\textstyle \Exp_B \left[\|\O_{\tilde F^{(B)}}\Pi_{B_1} \ket{\psi}\|_2\right]} \nonumber \\
& = & 2\cdot {\textstyle \Exp_B \left[\| \Pi_{B_1} \ket{\psi}\|_2\right]} \nonumber \\
%& \leq & 2 \sqrt{{\textstyle \Exp_B\left[ \rule{0pt}{8pt}|\bra{\psi} \Pi_{B_1} \ket{\psi}|\right]}},
& \leq & 2 \sqrt{{\textstyle \Exp_B\left[ \rule{0pt}{8pt}\|\Pi_{B_1} \ket{\psi}\|_2^2\right]}}, \label{eqn:blinding:bound-norm}
\end{eqnarray}
using Jensen's inequality in the last step. Let $\ket{\psi}=\sum_{x \in \bool^m, y \in \bool^n} \alpha_{x,y} \ket{x}\ket{y}$ where  $\|\ket{\psi}\|_2^2=\sum_{x,y} \alpha_{x,y}^2 = 1$.
Then
\begin{eqnarray*}
{\textstyle \Exp_B\left[\|\Pi_{B_1} \ket{\psi}\|_2^2\right]}
& = & {\textstyle \Exp_B\left[ \sum_{x,y: \, x \in B_1} \alpha_{x,y}^2 \right]} \nonumber\\
& = & \sum_{x,y} \alpha_{x,y}^2 \cdot \Pr[x \in B_1] \;\; \leq \;\; \epsilon. \label{eq:blinding:low-blinding-prob}
\end{eqnarray*}
Together with Equations~(\ref{eqn:adaptive1}) and (\ref{eqn:blinding:bound-norm}), this gives the desired result.
\qed

\subsection{Proof of the Resampling Lemma}\label{sec:perm-ARL}
%%%%%%%%%%%%%%%%%%%%%%%%%%%

%\ga{The below needs to be checked and aligned with the actual Lemma statement and the text surrounding it.}

%We now give the proof of \expref{Lemma}{lem:prp-arl}.

We begin by introducing a superposition-oracle technique based on the one by Zhandry~\cite{Zhandry2019}, but different in that our oracle represents a two-way accessible, uniform permutation (rather than a uniform function). We also do not need to ``compress'' the oracle, as an inefficient representation suffices for our purposes.
%\footnote{We do not need an efficient representation as we use the superposition oracle for the proof of \cref{lem:prp-arl}, which holds for unbounded (but query-limited) algorithms.}

%\medskip \noindent {\bf The superposition oracle.}
%%%

For an arbitrary function $f:\{0,1\}^n\to\{0,1\}^n$, define the state
\[
\ket{f}_F=\bigotimes_{x\in \bool^n}\ket{f(x)}_{F_x},
\]
where $F$ is the collection of registers $\{F_x\}_{x\in\{0,1\}^n}$. 
We represent an evaluation of $f$ via an operator~$O$ whose action on the computational basis is given by
\[
O_{XYF}\;\ket x_X\ket y_Y\ket f_F=\CNOT^{\otimes n}_{F_x:Y}\ket x_X\ket y_Y\ket f_F=\ket{x}_X\ket{y\oplus f(x)}_Y\ket f_F,
\]
where $X, Y$ are $n$-qubit registers. 
%Here the notation $\CNOT^{\otimes n}_{F_x:Y}$ indicates that the CNOT operators act on $F_x$ and $Y$, where $F_x$ is the control and $Y$ the target register.
%\jnote{We use $\mathcal{O}$ in the previous section, and use $O$ for other things throughout the paper; should we use different notation?} \cm{I need to think about this.}
Handling inverse queries to $f$ is more difficult.
We want to define an inverse operator $\Oinv$ such that, for any permutation~$\pi$, 
	\begin{equation}\label{eq:projected-inv-query}
\Oinv_{XYF} \ket \pi_F=\left(\sum_{x,y\in \bool^n}\proj{ y}_Y\otimes \mathsf X^{x}_X\otimes \proj{ y}_{F_{x}}\right) \ket{\pi}_F
\end{equation}
(where $\mathsf X$ is the Pauli-X operator, and for $x\in \bool^n$ we let $\mathsf X^{x} :=\mathsf X^{x_1}\otimes \mathsf X^{x_2}\otimes \ldots\otimes \mathsf X^{x_n}$ so that $\mathsf X^{x}\ket{\hat x}=\ket{\hat x\oplus x}$); then, 
\[\Oinv_{XYF} \ket x_X \ket y_Y \ket \pi_F=\ket {x \xor \pi^{-1}(y)}_X \ket{y}_Y \ket{\pi}_F.\]
%\jnote{Is there a reason to write $\Oinv_{YXF}$ rather than $\Oinv_{XYF}$?}
In order for $\Oinv$ to be a well-defined unitary operator, however, we must extend its definition to the entire space of functions.
A convenient extension is given by \extraexplain{the following action on arbitrary computational basis states:}
%\[
%\Oinv_{XYF}\ket x_X\ket y_Y\ket f_F=\prod_{x'\in \bool^n}\left(\mathsf{X}^{x'}_X\otimes \proj y_{F_{x'}}+\left(\mathds 1-\proj y\right)_{F_{x'}}\right)\ket x_X\ket y_Y\ket f_F,
%\]
\[
\Oinv_{XYF} =\prod_{x'\in \bool^n}\left(\mathsf{X}^{x'}_X\otimes \proj y_{F_{x'}}+\left(\mathds 1-\proj y\right)_{F_{x'}}\right),
\]
so that 
\[\Oinv_{XYF} \ket x_X \ket y_Y \ket f_F=\ket {x \xor \left(\oplus_{x': f(x')=y} \, x'\right)}_X \ket{y}_Y \ket{f}_F.\]
%
%\jnote{I don't see how the above matches the previous displayed equation. Should it be a tensor product rather than a product? }\cm{see my last comment.} \jnote{I guess I still don't see how the above product over $x'$ will reduce to a sum over $x', y'$ when we restrict $f$ to be a permutation.}\cm{If $f$ is a permutation and you apply the RHS of \cref{eq:projected-inv-query} to $\ket x_X\ket y_Y\ket f_F$, then only the summand $y'=y$, $x'=f^{-1}(y)$ survives. When applying the RHS prtoduct operator above, the only factor that acts non-trivially is the factor  $x'=f^{-1}(y)$. Now it is straightforward to compare the results.}
\extraexplain{In other words, the inverse operator XORs all preimages (under~$f$) of the value in register $Y$ into the contents of register~$X$.}

We may view a uniform permutation as a uniform superposition over all permutations in~$\Perms_n$; i.e., we model a uniform permutation as the state
\[
\ket{\phi_0}_F = \left(2^n!\right)^{-\frac 1 2}\sum_{\pi\in \Perms_n}\ket{\pi }_F.
\]
%\jnote{This notation seems to directly conflict with the notation $\ket{\pi}_F$. Can we use $\Phi_0$ without the bra-ket?} \cm{Do you mean that the conflict arises because $\phi_0$ is also a Greek letter?} \jnote{If $\phi_0$ is a function, then we already defined $\ket{\phi_0}_F$ to be the state corresponding to that function. I guess $\phi_0$ is not a function\ldots}\cm{right. I think the better way to clarify here is to emphasize that fact.}
The final state of any oracle algorithm $\D$ is identically distributed whether we (1)~sample uniform $\pi \in \Perms_n$ and then run $\D$ with access to~$\pi$ and $\pi^{-1}$, or (2)~run $\D$ with access to $O$ and $\Oinv$ after initializing the $F$-registers to $\ket{\phi_0}_F$ %, using $O$ and $\Oinv$ for its oracle queries, 
(and, if desired, at the end of its execution, measure the $F$-registers to obtain $\pi$ and the residual state of~$\D$).
Our proof relies on the following lemma, which is a special case of the conclusion %(i.e., the right-hand side) 
of implication~$(\diamond')$ in \cite{ODV21}. (Here and in the following, we denote the complementary projector of a projector $P$ by $\bar P\stackrel{\rm def}{=} \mathds 1-P$.)
\begin{lemma}[Gentle measurement lemma]\label{lem:sgentle}
	Let $\ket\psi$ be a quantum state and let $\{P_i\}_{i=1}^q$ be a collection of projectors %on some complex Euclidean space 
	with $\left\| \bar P_i\ket\psi\right\|_2^2\le \epsilon_i$ for all~$i$. Then 
	\[
	1-\left|\bra\psi\left(P_q \cdots P_1\right)\ket\psi\right|^2\le \sum_{i=1}^q\epsilon_i.
	\] 
\end{lemma}

\medskip \noindent {\bf Proof of Lemma~\ref{lem:prp-arl-proved}.}
We split the distinguisher $\D$ into two stages $\D =(\D_0, \D_1)$ corresponding to the first and second phases of the experiment in \expref{Lemma}{lem:prp-arl-proved}. 
As discussed above, we run the experiment using the superposition oracle $\ket{\phi_0}_F$ and then measure the $F$-registers at the end. 
Informally, our goal is to show that on average
over the choice of reprogrammed positions $s_0, s_1$, the adversary-oracle state after $\D_0$ finishes is almost invariant under the reprogramming operation (i.e., the swap of registers $F_{s_0}$ and $F_{s_1}$) unless $\D_0$ makes a large number of oracle queries. 
This will follow from \expref{Lemma}{lem:sgentle} because, on average over the choice of $s_0, s_1$, any particular query of $\D_0$ (whether using $O$ or $\Oinv$) only involves $F_{s_0}$ or $F_{s_1}$ with negligible amplitude. %\ga{Maybe one approach to breaking up Theorem 1 is to actually formally show this at this point?}

%\ga{leaving below early text here to help me sort out inconsistencies}
%Consider the following adaptive reprogramming game $G$ played by a two-stage adversary $\mathcal A=(\adver_0, \adver_1)$. $\adver_0$ gets quantum oracle access to a random permutation $\pi$ and its inverse, and outputs a quantum register $E$. Now two uniformly random inputs $x_0,x_1\in_R \bool^n$ and a random bit $b\in_R\{0,1\}$ are sampled. Define $\pi_0=\pi$ and 
%\begin{equation}
%		\pi_1(x)=\begin{cases}
%			\pi(x_{i\oplus 1})&\text{if }x=x_i\\
%			\pi(x)& \text{else.}
%		\end{cases}
%\end{equation} 
%Note that sampling two random inputs, two random outputs, or one random input  and one random output, are all equivalent. 
%Now $\adver_1$ gets input $(x_0, x_1, E)$ and oracle access to $\pi_b$ and outputs $b'$. $\adver$ wins if $b=b'$.
%\ga{end early text}

\def\S{{\sf Swap}}

%%%
%We are now ready to prove \expref{Lemma}{lem:prp-arl-proved}.
%\begin{theorem}\label{thm:prp-arl-proved}
%	The advantage that an algorithm $\adver=(\adver_0, \adver_1)$ can achieve in the adaptive reprogramming game $G$ is at most
%	\begin{equation}
%			\Pr[\adver \mathrm{\ wins}]-\frac 1 2\le 2\sqrt{\frac{q}{2^n}},
%	\end{equation}
%where $q$ is the total number of queries that $\adver_0$ makes to its oracles $\pi$ and $\pi^{-1}$.
%\end{theorem}
%\ga{Align this with above? I guess this is a restatement of ``arbitrary reprogramming lemma'' now.}\cm{It's a restatement of the "resampling lemma"}
%\begin{proof} 

We begin by defining the projectors 
\begin{eqnarray*}
&\left(P_{s_0s_1}\right)_X =\begin{cases}
\mathds 1 &s_0=s_1\\
\mathds 1-\proj{s_0}-\proj{s_1}& s_0\neq s_1
\end{cases}&\\
%and
%\[
&\left( P^{\mathrm{inv}}_{s_0s_1}\right)_{FY}=	 \begin{cases}
\mathds 1 &s_0=s_1\\
\sum_{y\in \bool^n}\proj y_Y\otimes \left(\mathds 1-\proj y\right)^{\otimes 2}_{F_{s_0}F_{s_1}}& s_0\neq s_1.
\end{cases}&
\end{eqnarray*}
%\jnote{I'm confused by the notation above. Is 	\[\left(\mathds 1-\proj y\right)^{\otimes 2}_{F_{s_0}F_{s_1}} = 	\left(\mathds 1-\proj y\right)_{F_{s_0}} \otimes \left(\mathds 1-\proj y\right)_{F_{s_1}}?\]
%If so, the notation is confusing. If not, the notation is even more confusing.}\cm{yes, so at least it's only confusing}
%\jnote{Can you give intuition for what the second operator is doing?}\cm{it projects onto the subspace where the adversary's inverse query input is not equal to the content of $F_{x_i}$ for $i=0,1$. On this subspace, the resampling cannot be noticed, because the resampled input-output pairs are not involved in the query.}
It is straightforward to verify that for any $s_0, s_1$:
\begin{eqnarray}
&\left[\S_{F_{s_0}F_{s_1}},\; O_{XYF}\left(	 P_{s_0s_1}\right)_X\right]=0 &\label{eq:commutes1}\\
&\left[\S_{F_{s_0}F_{s_1}},\; \Oinv_{XYF}\left( P^{\mathrm{inv}}_{s_0s_1}\right)_{FY} \right]=0,& \label{eq:commutes2}
\end{eqnarray}
where $[\cdot, \cdot]$ denotes the commutator operation, %$P^\Per$ is the projector onto the subspace of $F$ spanned by all functions that are \emph{permutations}, 
and
$\S_{AB}$ is the swap operator (i.e., $\S_{A,B} {\ket x}_A\ket{x'}_B=\ket{x'}_A\ket{x}_B$ if the target registers $A, B$ are distinct, and the identity if $A$ and $B$ refer to the same register). 
In words, this means that if we project a forward query to inputs other than $s_0, s_1$, then swapping the outputs of a function at $s_0$ and $s_1$ before evaluating that function has no effect; the sane holds if we project an inverse query (for some associated function~$f$) to the set of output values that are not equal to $f(s_0)$ or~$f(s_1)$.

Since $\bar P_{s_0s_1}\stackrel{\rm def}{=}\mathds 1- P_{s_0s_1} \le\proj{s_0}+\proj{s_1}$  it follows that
for any normalized state $\ket\psi_{XE}$ (where $E$ is an arbitrary other register),
\begin{align}
\Exp_{s_0, s_1}\left[\left\|\left(\bar P_{s_0s_1}\right)_{X}\ket \psi_{XE}\right\|^2_2\right]
&\le \Exp_{s_0, s_1}\left[\bra\psi\left(\proj{s_0}+\proj{s_1}\right)\ket \psi\right] \nonumber \\
&= 2\cdot 2^{-n}.\label{eq:small}
\end{align}
We show a similar statement about $P^{\mathrm{inv}}_{s_0s_1}$. We can express a valid adversary/oracle state $\ket\psi_{YXEF}$~(that is thus only supported on the span of $\Perms_n$ \jnote{???}) as
% be the query state of an adversary interacting with the inverse oracle. We can write %\cm{$c_y$ depends on $\pi$}
\begin{equation}\label{eq:decomposition}
\ket\psi_{YXEF}=\sum_{x,y\in \bool^n}c_{xy}\ket y_Y\ket y_{F_x}\ket{\psi_{xy}}_{XEF_{x^c}},
\end{equation}
\jnote{Why do you need $\ket y_Y$, rather than just pushing it inside $\psi_{xy}$?}
for some normalized quantum states $\{\ket{\psi_{xy}}\}_{x,y\in\bool^n}$,
with $\sum_{x,y\in \bool^n}|c_{xy}|^2=1$ and $\bra y_{F_{x'}}\ket{\psi_{xy}}_{XEF_{x^c}}=0$ for all $x'\neq x$.
%.
%Suppose $s_0\neq s_1$. The projector $ P^{\mathrm{inv}}_{s_0,s_1}$ is diagonal in the computational basis,~so 
%\begin{align*}
%&\Exp_{s_0, s_1}\left[	\left\|\left(\bar P^{\mathrm{inv}}_{s_0s_1}\right)_{YF}\ket\psi_{YXEF}\right\|_2^2\right]\\
%&=\Exp_{s_0, s_1}\left[\bra \psi_{YXEF}\left(\bar P^{\mathrm{inv}}_{s_0s_1}\right)_{YF}\ket\psi_{YXEF}\right]\\
%&=\sum_{y\in \bool^n}\sum_{\pi\in \Perms_n}|c_{y\pi}|^2\left(2^n!\right)^{-1}\bra\pi_F\left(\mathds 1-\left(\mathds 1-\proj y\right)^{\otimes 2}_{F_{s_0}F_{s_1}}\right)\ket\pi_F\\
%&=\sum_{y\in \bool^n}\sum_{\substack{\pi\in \Perms_n:\\ y\in\{\pi(s_0), \pi(s_1)\}}}|c_{y\pi}|^2\left(2^n!\right)^{-1}\\
%&=\sum_{\substack{\pi\in \Perms_n:\\ y\in\{\pi(s_0), \pi(s_1)\}}}\left(2^n!\right)^{-1}\\
%&=\frac{2\left(2^n-1\right)!}{2^n!}\;\; = \;\; 2\cdot 2^{-n}.
%\end{align*}
%%\ga{Not sure what happened in the third line. How did we collapse to a single $\pi$? I was expecting something like $\sum_{\pi, \pi'} \bra{\pi}_F(\text{stuff})_F \ket{\pi'}_F$.}
%%(The fourth inequality uses the normalization of the coefficients $c_{y\pi}$.) 
If $s_0=s_1$, then $\left\|\left(\bar P^{\mathrm{inv}}_{s_0s_1}\right)_{YF}\ket\psi_{YXEF}\right\|_2^2=0\le 2\cdot 2^{-n}$. It is thus immediate from \cref{eq:decomposition} that
\begin{align}\label{eq:invsmall}
&\Exp_{s_0, s_1}\left[\left\|\left(\bar P^{\mathrm{inv}}_{s_0s_1}\right)_{YF}\ket\psi_{YXEF}\right\|_2^2\right]	\le 2\cdot 2^{-n}
\end{align}
%for all $s_0, s_1$.
\jnote{I am missing why the above is not immediate, as it was for the forward direction.}\cm{I guess the problem is that I am missing why it would be immediate... ;) I simplified it, but felt it was not obvious enough without giving the (new) decomposition \cref{eq:decomposition}.} \jnote{Somehow to me the explanation doesn't convince me, but I am anyway convinced without the explanation that it is true.} 
%In the following, we model $\adver_0$'s oracle as a superposition oracle. 

Without loss of generality, we assume $\D_0$ starts with initial state $\ket{\psi_0}=\ket{\psi_0'}\ket{\phi_0}$ (which we take to include the superposition oracle's initial state $\ket{\phi_0}$), computes the state
\[
\ket\psi=U_{\D_0}\ket{\psi_0}=U_qO_qU_{q-1}O_{q-1}\cdots U_1O_1\ket{\psi_0},
\]
and outputs all its registers as a state register~$E$. Here, each $O_i\in\{O, \Oinv\}$ acts on registers $XYF$, and each $U_j$ acts on registers $XYE$. To each choice of $s_0,s_1$ we assign a decomposition $\ket\psi=\ket{\psi_{\mathrm{good}}(s_0,s_1)}+\ket{\psi_{\mathrm{bad}}(s_0,s_1)}$ by defining

\[
\ket{\psi_{\mathrm{good}}(s_0,s_1)}=z\cdot U_qO_q P^q_{s_0s_1}U_{q-1}O_{q-1} P^{q-1}_{s_0s_1}\cdots U_1O_1 P^1_{s_0s_1}\ket{\psi_0},
\]
where $ P^i_{s_0s_1}= P_{s_0s_1}$ if $O_i=O$, $ P^i_{s_0s_1}= P^{\mathrm{inv}}_{s_0s_1}$ if $O_i=\Oinv$, and $z\in\mathbb C$ is such that $|z|=1$ and   $\braket\psi{\psi_{\mathrm{good}}(s_0,s_1)}\in\mathbb R_{\ge 0}$. \jnote{Is the correction factor $z$ necessary?}\cm{I added a parenthetical to explain.} \jnote{I don't really see why it is needed, but I guess it doesn't matter.}
\[
\ket{\psi_{\mathrm{good}}(s_0,s_1)}=z\cdot U_{\D_0} %\left(\prod_{i=q}^1 
Q^q_{s_0s_1}\cdots Q^1_{s_0s_1}\ket{\psi_0},
\]
with
$Q^i_{s_0s_1}=\tilde U_i^\dagger P^i_{s_0s_1}\tilde U_i$
for 
$\tilde U_i=U_{i-1}O_{i-1}\ldots U_1O_1$.
Let 
\[
\epsilon_i(s_0,s_1)=\left\| \bar Q^i_{s_0s_1}\ket{\psi_0}\right\|_2^2=\left\| \bar P^i_{s_0s_1}\tilde U_i\ket{\psi_0}\right\|_2^2.
\]
Applying Lemma~\ref{lem:sgentle} yields
\begin{eqnarray}
1-\left|\braket\psi{\psi_{\mathrm{good}}(s_0,s_1)}\right|^2\le \sum_{i=1}^q	\epsilon_i(s_0,s_1). \label{eqn:gentle}
\end{eqnarray}

We will now analyze the impact of reprogramming the superposition oracle after $\D_0$ has finished. Recall that reprogramming swaps the values of the permutation at points $s_0$ and $s_1$, which is implemented in the superposition-oracle framework by applying $\S_{F_{s_0}F_{s_1}}$. 
Note that 
%the superposition oracle $\ket{\phi_0}$ is invariant under the action of $\Perms_n$ that permutes the registers $F_x$. In particular, 
$\S_{F_{s_0}F_{s_1}}\ket{\phi_0}=\ket{\phi_0}$. As the adversary's internal unitaries $U_i$ do not act on $F$, Equations \eqref{eq:commutes1} and \eqref{eq:commutes2} then imply that 
%\jnote{I don't immediately see why} \cm{The sentence should have referred to both equations \eqref{eq:commutes1} and \eqref{eq:commutes2}, I changed that. The key is that in $\ket{\psi_{\mathrm{good}}(s_0,s_1)}$, the oracle unitaries appear only in combination with the projectors that project onto the subspace where the respective oracle unitary does not touch the $F_{s_i}$.} %$S_{F_{s_0}F_{s_1}}$
\[
\S_{F_{s_0}F_{s_1}}\ket{\psi_{\mathrm{good}}(s_0,s_1)}=\ket {\psi_{\mathrm{good}}(s_0,s_1)}\,.
\]
%as the superposition oracle's initial state $\ket{\phi_0}$ is invariant under the action of $\Perms_n$ that permutes the registers $F_x$, so in particular $S_{F_{s_0}F_{s_1}}\ket{\phi_0}=\ket{\phi_0}$. 
The standard formula for the trace distance of pure states thus yields
\begin{align}
\frac 1 2\left\|\proj\psi-\S_{F_{s_0}F_{s_1}}\proj{\psi}\S_{F_{s_0}F_{s_1}}\right\|_1=&\sqrt{1-\left|\bra\psi \S_{F_{s_0}F_{s_1}}\ket\psi\right|^2}.\label{eq:puretrdist}
%	\le&\Big\|\ketbra{\psi_{\mathrm{good}}(x_0,x_1)}{\psi_{\mathrm{bad}}(x_0,x_1)}\left(\mathds 1-S_{F_{x_0}F_{x_1}}\right)\nonumber\\
%	&\quad+\left(\mathds 1-S_{F_{x_0}F_{x_1}}\right)\ketbra{\psi_{\mathrm{bad}}(x_0,x_1)}{\psi_{\mathrm{good}}(x_0,x_1)}\nonumber\\
%	&\quad+S_{F_{x_0}F_{x_1}}\proj{\psi_{\mathrm{bad}}(x_0,x_1)}S_{F_{x_0}F_{x_1}}\Big\|_1\nonumber\\
%	\le&2\Big\|\ketbra{\psi_{\mathrm{good}}(x_0,x_1)}{\psi_{\mathrm{bad}}(x_0,x_1)}\left(\mathds 1-S_{F_{x_0}F_{x_1}}\right)\Big\|_1\nonumber\\
%	&\quad+\Big\|S_{F_{x_0}F_{x_1}}\proj{\psi_{\mathrm{bad}}(x_0,x_1)}S_{F_{x_0}F_{x_1}}\Big\|_1\nonumber\\
%	\le&4\|\ket{\psi_{\mathrm{bad}}(x_0,x_1)}\|_2+\|\ket{\psi_{\mathrm{bad}}(x_0,x_1)}\|_2^2
\end{align}
We further have
\begin{align}
\left|\bra\psi \S_{F_{s_0}F_{s_1}}\ket\psi\right|
&= \left|\braket\psi\psi+ \bra{\psi_{\mathrm{bad}}(s_0,s_1)}\left(\S_{F_{s_0}F_{s_1}}-\mathds 1\right)\ket{\psi_{\mathrm{bad}}(s_0,s_1)}\right|\nonumber \\
&\ge 1-2\|\ket{\psi_{\mathrm{bad}}(s_0,s_1)}\|_2^2\label{eq:use-inva}
\end{align}
using the triangle and Cauchy-Schwarz inequalities.
Combining Equations \eqref{eq:puretrdist} and \eqref{eq:use-inva} we obtain
\begin{equation*}
\frac 1 2\left\|\proj\psi-\S_{F_{s_0}F_{s_1}}\proj{\psi}\S_{F_{s_0}F_{s_1}}\right\|_1\le 2\cdot \|\ket{\psi_{\mathrm{bad}}(s_0,s_1)}\|_2.
\end{equation*}
But as $\ket{\psi_{\mathrm{bad}}(s_0,s_1)}=\ket \psi-\ket{\psi_{\mathrm{good}}(s_0,s_1)}$, we have
\begin{align}
\|\ket{\psi_{\mathrm{bad}}(s_0,s_1)}\|_2^2
&= 2-2\cdot \mathrm{Re}\braket{\psi}{\psi_{\mathrm{good}}(s_0,s_1)}\nonumber\\
&= 2-2\cdot \left|\braket{\psi}{\psi_{\mathrm{good}}(s_0,s_1)}\right|\label{eq:use-z}\\
&\le 2\sum_{i=1}^q\epsilon_i(s_0,s_1).\nonumber
\end{align}
Combining the last two equations we obtain
\begin{equation}\label{eq:tracegentle}
\frac 1 2\left\|\proj\psi-\S_{F_{s_0}F_{s_1}}\proj{\psi}\S_{F_{s_0}F_{s_1}}\right\|_1\le 2\sqrt 2\sqrt{\sum_{i=1}^q\epsilon_i(s_0,s_1)}\,.
\end{equation}
The remainder of the proof is the same as the analogous part of the proof of \cite[Theorem~6]{GHHM20}. $\D_1$'s task boils down to distinguishing the states $\ket\psi$ and $\S_{F_{s_0}F_{s_1}}\ket{\psi}$, for uniform $s_0, s_1$ that $\D_1$ receives as input, using the limited set of instructions allowed by the superposition oracle. 
We can therefore bound $\D$'s advantage by the maximum distinguishing advantage for these two states when using arbitrary quantum computation, averaged over the choice of $s_0, s_1$. Using the standard formula for this maximum distinguishing advantage we obtain
\begin{align*}
\Pr\left[\D \text{~outputs~} b\right]-\frac 1 2
&\le \frac 1 4 \Exp_{s_0,s_1}\left[	\left\|\proj\psi-\S_{F_{s_0}F_{s_1}}\proj{\psi}\S_{F_{s_0}F_{s_1}}\right\|_1\right]\\
&\le \sqrt 2\Exp_{s_0,s_1}\left[	\sqrt{\sum_{i=1}^q\epsilon_i(s_0,s_1)}\right]\\
&\le \sqrt 2\sqrt{\Exp_{s_0,s_1}\left[	\sum_{i=1}^q\epsilon_i(s_0,s_1)\right]}
\le 2\sqrt{\frac{q}{2^n}},
\end{align*}
where the second inequality is Equation \eqref{eq:tracegentle}, the third  is Jensen's inequality, and the last is from Equations~\eqref{eq:small}--\eqref{eqn:gentle}. %\eqref{eq:invsmall}, and \eqref{eqn:gentle}. 
This implies the lemma.
\qed
%\end{proof}

\section*{Acknowledgments}
The authors thank Andrew Childs and Bibhusa Rawal for useful discussions.
Work of Jonathan Katz was supported in part 
by financial assistance award 70NANB19H126 from the U.S. Department of Commerce, National Institute of Standards and Technology. Work of Christian Majenz was funded by a NWO VENI grant (Project No.\ VI.Veni.192.159). Gorjan Alagic acknowledges support from the U.S. Army Research Office under Grant Number W911NF-20-1-0015, the U.S. Department of Energy under Award Number DE-SC0020312, and the AFOSR under Award Number FA9550-20-1-0108.
%Omitted due to anonymity. \jnote{Put back later.}

\def\shortbib{0}

%%%%%%%%%%%%%%%%%%%%%%%%%%%
%%%%%%%%%%%%%%%%%%%%%%%%%%%

%%%%%%%%%%%%%%%%%%%%%%%%%%%
%%%%%%%%%%%%%%%%%%%%%%%%%%%
\appendix
%%%%%%%%%%%%%%%%%%%%%%%%%%%
%%%%%%%%%%%%%%%%%%%%%%%%%%%

%%%%%%%%%%%%%%%%%%%%%%%%%%%
\section{Security of Forward-Only Even-Mansour}\label{app:forward-EM}
%%%%%%%%%%%%%%%%%%%%%%%%%%%

%\ga{Moved this to appendix. We can always resurrect it later if we feel it's needed, e.g., for Elephant.}

%\jnote{Unclear what is happening with this section. If we are keeping it, then it probably needs a pass.}

In this section we consider a simpler case, where $E_k[F](x) := F(x \oplus k)$ for $F:\bool^n \rightarrow \bool^n$ a uniform \emph{function} and $k$ a uniform $n$-bit string.
Here we restrict the adversary to forward queries only, i.e., the adversary has classical access to $E_k[F]$ and quantum access to~$F$; note that $E^{-1}_k[F]$ and $F^{-1}$ may not even be well-defined. As mentioned in the main body, this setting was analyzed in the previously published work \cite{JST21} as well, using different techniques.

We let $\Funcs_n$ denote the set of all functions from $\bool^n$ to~$\bool^n$.

\begin{theorem}\label{thm:forward-only}
%	\jnote{Change $P$ to $F$ throughout if we are letting it be an arbitrary function.}\chen{Done.}
Let $\A$ be a quantum algorithm making $\formerqC$  classical queries to its first oracle and $q_F$ quantum queries to its second oracle. Then
\begin{eqnarray*}
\lefteqn{\left|\Pr_{\substack{k \from \bool^n \\ F \from \Funcs_n}} \left[\A^{E_k[F], F}(1^n) = 1\right]
- \Pr_{R, F \from \Funcs_n} \left[\A^{R, F}(1^n) = 1\right]\right|} \hspace*{1.5in} \\
& \leq & 2^{-n/2} \cdot \left(2 \formerqC \sqrt{q_F} + 2q_F \sqrt{\formerqC}\right).
\end{eqnarray*}
%\jnote{What do you think of the notation $A^{O, \ket{O'}}$ to denote that $A$ has classical access to $O$ and quantum access to~$O'$?} \ga{It's okay with me. I think Zhandry used this in the random functions paper.}\chen{I agree with that too, but since we use the same notation in section 3, if we change this we should then change both.}
\end{theorem}

%\subsection{Security proof}
%%%%%%%%%%%%%%%%%%%%%%%%%%%
\begin{proof}
We make the same assumptions about $\A$ as in the initial paragraphs of the proof of \expref{Theorem}{thm:full}. We also adopt analogous notation for  the stages of $\A$, now using $q_E$, $q_F$, and $q_{F, j}$ as appropriate.

Given a function $F : \bool^n \rightarrow \bool^n$, a set $T$ of pairs where any $x \in \bool^n$ is the first element of at most one pair in~$T$, and a key $k \in \bool^n$, we define the function $F_{T, k}:\bool^n\rightarrow \bool^n$ as
\[
F_{T, k}(x) := 
\begin{cases}
y &\text{if } (x \oplus k, y) \in T\\
F(x) &\text{otherwise.}
\end{cases}
\]
Note that, in contrast to the analogous definition in \expref{Theorem}{thm:full}, here the order of the tuples in $T$ does not matter and so we may take it to be a set. Note also that we are redefining the notation $F_{T, k}$ from how it was used in \expref{Theorem}{thm:full}; this new usage applies to this Appendix only.

%\ga{As Jon pointed out, this is the same notation that we use for a different (but analogous) notion in the two-way case. We should either clarify that clearly here, or use different notation in this section.}\chen{I think this is fine if I change $P$ to $F$.}
We now define a sequence of experiments $\Hyb_j$, for $j=0, \ldots, \formerqC$:

\medskip\noindent \textbf{Experiment} $\Hyb_j$. Sample $R, F \from \Funcs_n$ and $k \from \bool^n$. Then:
\begin{enumerate}
\item Run $\A$, answering its classical queries using~$R$ and its quantum queries using~$F$, stopping immediately before its $(j+1)$st classical query. Let $T_j = \{(x_1, y_1), \dots, (x_j, y_j)\}$ be the set of all classical queries made by $\A$ thus far and their corresponding responses.
\item For the remainder of the execution of $\A$, answer its classical queries using $E_k[F]$ and its quantum queries using $F_{T_j, k}$. 
\end{enumerate}
We can represent $\Hyb_j$ as the experiment in which $\A$'s queries are answered using the oracle sequence
\[
\underbrace{F, R, F, \cdots, R, F}_{\mbox{\scriptsize $j$ classical queries}}, \underbrace{E_k[F], F_{T_j, k}, \cdots, E_k[F], F_{T_j, k}}_{\mbox{\scriptsize $\formerqC-j$ classical queries}}\,.
\]
Note that $\Hyb_0$ is exactly the real world (i.e., $\A^{E_k[F], F}$) and $\Hyb_{\formerqC}$ is exactly the ideal world (i.e., $\A^{R, F}$.)
%\[\Pr[\A(\Hyb_0)=1] = \Pr_{\substack{F \from \Funcs_n\\k \from \bool^n}} \left[\A^{E_k[F], F}(1^n) = 1\right]\]
%and
%\[\Pr[\A(\Hyb_{\formerqC})=1] = \Pr_{\substack{F \from \Funcs_n\\R \from \Perms_n}} \left[\A^{R, F}(1^n) = 1\right].\]
%$\Hyb_0$ corresponds to the case where $\A$ is run with all its classical queries answered using $E_k[P]$ and all its quantum queries answered using~$P$, whereas $\Hyb_{\formerqC}$ corresponds to the case where $\A$ is run with all its classical queries answered using $R$ and all its quantum queries answered using~$P$.

For $j=0, \ldots, \formerqC-1$, we define an additional experiment~$\Hyb_j'$:

\medskip\noindent \textbf{Experiment} $\Hyb_j'$. 
Sample $R, F \from \Funcs_n$ and $k \from \bool^n$. Then:
\begin{enumerate}%[label=(\Roman*)]
\item Run $\A$, answering its classical queries using $R$ and its quantum queries using~$F$, stopping immediately after its $(j+1)$st classical query. Let $T_{j+1} = \big((x_1, y_1), \dots, (x_{j+1}, y_{j+1})\big)$ be the set of all classical queries made by $\A$ thus far and their corresponding responses.
\item For the remainder of the execution of $\A$, answer its classical queries using $E_k[F]$ and its quantum queries using~$F_{T_{j+1}, k}$.
\end{enumerate}
I.e., $\Hyb'_j$ corresponds to answering $\A$'s queries using the oracle sequence
\[
\underbrace{F, R, F, \cdots, R, F}_{\mbox{\scriptsize $j$ classical queries}}, R, F_{T_{j+1}, k}, \underbrace{E_k[F], F_{T_{j+1}, k} \cdots, E_k[F], F_{T_{j+1}, k}}_{\mbox{\scriptsize $\formerqC-j-1$ classical queries}}\,.
\]

%\paragraph{Bounding the distinguishability of the hybrids.}
%%%

We now show that $\Hyb_j'$ is close to $\Hyb_{j+1}$ and $\Hyb_j$ is close to~$\Hyb_j'$ for $0 \leq j < q_E$.

\begin{lemma}\label{lem:step-two}
	For $j=0, \ldots, q_E-1$, 
\[|\Pr[\A(\Hyb'_{j})=1] - \Pr[\A(\Hyb_{j+1})=1] | \leq 2\cdot q_{F,j+1}
\sqrt{(j+1)/2^n}.\]
%where $q_{F,j+1}$ is the expected number of queries $\A$ makes to $F$ between the $(j+1)$st and the $(j+2)$nd query in the ideal world. 
\end{lemma}
\begin{proof}
%Recall we can write the oracle sequences defined by $\Hyb_{j}'$ and $\Hyb_{j+1}$~as
%\begin{alignat*}{5}
%\Hyb_{j}': \;\; &F,R, F, \cdots, R, F, ~~&&R,\,&&F_{T_{j+1}, k}, ~~&&\,E_k[F], F_{T_{j+1}, k}, \cdots, E_k[F], F_{T_{j+1}, k} \\
%
%\Hyb_{j+1}: \;\; &\underbrace{F, R, F, \cdots, R, F,}_{\mbox{\scriptsize $j$ classical queries}} ~~&&R,\,&&F, &&\underbrace{E_k[F], F_{T_{j}, k}, \cdots, E_k[F], F_{T_{j}, k}}_{\mbox{\scriptsize $\formerqC-j$ classical queries}}\,.
%\end{alignat*}
Given an adversary $\A$, we construct a distinguisher $\D$  for the ``blinding game'' of \expref{Lemma}{lem:ada-blind}
that works as follows: 
%Recall that $\algo D$ needs to provide a function $F$ and a blinding algorithm $\algo B$ to $\algo C$ before the blinding game can begin. We first describe how $\algo D$ will generate $F$ and $\algo B$ during a ``setup phase,'' and then describe the behavior of $\algo D$ during the actual blinding game.
\begin{description}
\item[Phase 1:] 
 $\D$ samples $F, R \leftarrow \Funcs_n$. It then
runs $\A$, answering its quantum queries with $F$ and its classical queries with~$R$, until it replies to $\A$'s $(j+1)$st classical query. Let $T_{j+1} = \{(x_1, y_1), \ldots, (x_{j+1}, y_{j+1})\}$ be the set of classical queries made by $\A$ and their responses. $\D$ defines 
algorithm $\algo B$ as follows: on randomness~$k \in \bool^n$, output $B=\{(x_j \oplus k, y_j)\}_{j=1}^{j+1}$. 
Finally, $\D$ outputs $F$ and $\mathcal{B}$.
%, also define $B_1=\{x\oplus k: (x,y)\in T_j\}$. 
\item[Phase 2:] $\D$ is given quantum access to a function~$F_b$. It continues to run~$\A$, answering its quantum queries with~$F_b$ until $\A$ makes its next 
%(i.e., $(j+1)$st) 
classical query.
\item[Phase 3:] $\D$ is given the randomness~$k$ used to run~$\B$. It continues 
running $\A$, answering its classical queries with $E_k[F]$ and its quantum queries with~$F_{T_{j+1}, k}$. Finally, $\D$ outputs whatever $\A$ outputs.
\end{description}

When $b=0$ (so $F_b=F_0=F$), then $\A$'s output is identically distributed to its output in $\Hyb_{j+1}$.
On the other hand, when $b=1$ then $F_b=F_1=F^{(B)} = F_{T_{j+1},k}$ and so $\A$'s output is identically distributed to its output in~$\Hyb'_j$.
The expected number of queries made by $\D$ in phase~2 when $F=F_0$ is the expected number of queries made by $\A$ in stage $(j+1)$ in $\Hyb_{j+1}$. Since $\Hyb_{j+1}$ and $\Hyb_{q_E}$ are identical until after the $(j+1)$st stage, this is precisely~$q_{F,j+1}$. 
Because $k$ is uniform, we can apply \expref{Lemma}{lem:ada-blind} with $\epsilon=(j+1)/2^n$. The lemma follows.
\qed\end{proof}

%We can strengthen the previous result to show the following:

%\begin{lemma}\label{lem:step-2.5}
%	\[\sum_{j=1}^{\formerqC} \left( \Pr[\A(\Hyb'_{j-1})=1] - \Pr[\A(\Hyb_{j})=1] \right) \leq 2q_F \sqrt{\formerqC/2^n}.\]
%\end{lemma}

%\begin{proof}
%Examining the proof of Lemma~\ref{lem:step-two} shows that for any $R, F$ we have
%\[\Pr[\A(\Hyb'_{j-1})=1] - \Pr[\A(\Hyb_{j})=1]  \leq 2q_{F,j} \sqrt{\formerqC/2^n},\]
%where the probabilities are only over choice of~$k$ and the randomness of any %quantum measurements made by~$\A$, and
%where $q_{F,j}$ denotes the number of queries made by~$\A$ in stage~$j$ when %its oracles are those of experiment~$\Hyb_j$. (Although $q_{F,j}$ may depend on~$R, F$, it is independent of $k$ and any measurements made by~$\A$. \jnote{When I wrote this I was thinking of measurements made by $\A$ before phase~$j$. But I guess it is not true for measurements made by $\A$ during phase~$j$.} Thus, $q_{F,j}$ is well defined once $R, F$ are fixed.)
%Note that $q_{F,j}$ is also the number of queries made by $\A$ in stage~$j$ when its oracles are those of experiment~$\Hyb_{\formerqC}$, since the oracles of $\Hyb_j$ and the oracles of $\Hyb_{\formerqC}$ are identical up to stage~$j$. This means that $\sum_{j=0}^{\formerqC} q_{F,j} = q_F$, and the lemma follows.
%\qed
%\qed\end{proof}

\begin{lemma}\label{lem:step-one}
	%$\Hyb_j \approx \Hyb_j'$ 
	For $j=0, \ldots, q_E$,
	\[|\Pr[\A(\Hyb_j) = 1] - \Pr[\A(\Hyb'_j)=1]| \leq 1.5 \cdot \sqrt{q_F/2^n}\,.\]
	%\jnote{The advantage is 0 when $j=0$. I think that is correct (i.e., the hybrids are identical in that case), but it was tricky to verify.}
\end{lemma}
\begin{proof}
%	Recall we can write the oracle sequences defined by $\Hyb_j$ and $\Hyb_j'$~as
%	\begin{alignat*}{5}
%	\Hyb_j: \;\; &F, R, F, \cdots; R, F, ~~&&E_k[F], &&\; F_{T_j, k}; ~~~&&E_k[F], F_{T_j, k}~~\,, \cdots, E_k[F], F_{T_j, k} \\
%	\Hyb'_j: \;\; &\underbrace{F, R, F, \cdots, R, F}_{\mbox{\scriptsize $j$ classical queries}}, ~~&&R, &&F_{T_{j+1}, k}, &&\underbrace{E_k[F], F_{T_{j+1}, k}, \cdots, E_k[F], F_{T_{j+1}, k}}_{\mbox{\scriptsize $\formerqC-j-1$ classical queries}}\,.
%	\end{alignat*}
	%Let $\A$ be a quantum algorithm that distinguishes $\Hyb_j$ from $\Hyb_j'$ with some probability. 
	From any adversary $\A$, we construct a distinguisher $\algo D$ for the game of \expref{Lemma}{lem:arl}. $\D$ works as follows:
	\begin{description}
		\item[Phase 1:] $\D$ is given quantum access to a (random) function~$F$. It samples $R \from \Funcs_n$ and
		then runs $\A$, answering its quantum queries using~$F$ and its classical queries using~$R$, until $\A$ submits its $(j+1)$st classical query~$x_{j+1}$. At that point, let $T_j=\{(x_1,y_1) \ldots, (x_j, y_j)\}$ be the set of input/output pairs $\A$ has received from its classical oracle thus far.
		\item[Phase 2:] $\D$ is given (uniform) $s \in \bool^n$ and quantum oracle access to a function~$F_b$. 
		$\D$ sets $k := s \oplus x_{j+1}$, and then
		continues running $\A$, answering its classical queries (including the $(j+1)$st) using $E_k[F_b]$ and its quantum queries using the function $(F_b)_{T_j, k}$, i.e., 
		\[
		x \mapsto 
		\begin{cases}
		y &\text{if } (x \oplus k, y) \in T_j \\
		F_b(x) &\text{otherwise.}
		\end{cases}
		\]
		Finally, $\D$ outputs whatever $\A$ outputs.
	\end{description}
	
	We analyze the execution of $\algo D$ in the two cases of the game of \expref{Lemma}{lem:arl}. 
	In either case, the quantum queries of $\A$ in stages $0, \ldots, j$
	are answered using a random function~$F$, and $\A$'s first $j$ classical queries are answered using an independent random function~$R$. Note further that since $s$ is uniform, so is~$k$. 
	
	%in the two relevant cases: $b=0$ (no reprogramming) and $b=1$ (reprogramming.) Both cases begin with $\algo C$ sampling a random function $F$ and $\algo D$ sampling a random function $G$. We relabel $F$ to $P$ and $G$ to $R$. Next, $\A$ is run with quantum oracle $O_0 = P$ and classical oracle $R$ until just prior to its $(j+1)$-st classical query, which is $x^*$. Let $L=(x_1, \cdots, x_j)$ be its first $j$ classical queries. Note that, so far, this execution is consistent with both $\Hyb_j$ and $\Hyb_j'$.
	
	%Next, $\algo C$ samples a uniform value $s$, and $\algo D$ sets $k = s \oplus x^*$; note $k$ is then also uniformly random and independent. The two cases now begin to diverge.\\
	
	\medskip\noindent \textbf{Case 1: $b=0$.} 
	In this case, all the remaining classical queries of~$\A$ (i.e., from the $(j+1)$st on) are answered using 
	$E_k[F]$, and the remaining quantum queries of $\A$ are answered using $F_{T_j, k}$. The output of $\A$ is thus distributed identically to its output in~$\Hyb_j$ in this case.
	
	\medskip\noindent \textbf{Case 2: $b=1$.} Here, $F_b=F_1=F_{s \rightarrow y}$ for a uniform~$y$.
	Now, the response to the $(j+1)$st classical query of $\A$ is 
	\[E_k[F_b](x_{j+1}) = E_k[F_{s \rightarrow y}](x_{j+1}) = F_{s \mapsto y}(k \oplus x_{j+1}) = F_{s \rightarrow y}(s) = y.\] 
	Since $y$ is uniform and independent of anything else, and since $\A$ has never previously queried $x_{j+1}$ to its classical oracle, this is equivalent to answering the first $j+1$ classical queries of~$\A$ using a random function~$R$.
	The remaining classical queries of $\A$ are also answered using $E_k[F_{s \mapsto y}]$. However, since $E_k[F_{s \rightarrow y}](x)=E_k[F](x)$ for all $x \neq x_{j+1}$ and $\A$ never repeats the query~$x_{j+1}$, this is equivalent to answering the remaining classical queries of $\A$ using~$E_k[F]$.
	
	The remaining quantum queries of $\A$ are answered with the function
	\[
	x \mapsto 
	\begin{cases}
	y' &\text{if } (x \oplus k, y') \in T_j\\
	F_{s \rightarrow y}(x)% = 
	%\begin{cases}
	%y &\text{if } x = s \Leftrightarrow x \oplus k=x^* \\
	% F(x) &\text{otherwise}
	%\end{cases}
	&\text{otherwise.}
	\end{cases}
	\]
	This, in turn, is precisely the function $F_{T_{j+1}, k}$, where $T_{j+1}$ is obtained by adding $(x_{j+1}, y)$ to $T_j$ (and thus
	consists of the first $j+1$ classical queries made by $\A$ and their corresponding responses). Thus, the output of $\A$ in this case is distributed identically to its output in~$\Hyb_j'$.
			
	The number of quantum queries made by $\D$ in phase~1 is at most~$q_F$. The claimed result thus follows from \expref{Lemma}{lem:arl}. 
	\qed\end{proof}

Using Lemmas~\ref{lem:step-two} and~\ref{lem:step-one}, and the fact that $\sum_{j=1}^{\formerqC}q_{F,j}=q_F$, we have
\begin{eqnarray*}
%\lefteqn{\left|\Pr_{\substack{F \from \Funcs_n\\k \from \bool^n}} \left[\A^{E_k[F], F}(1^n) = 1\right]
%- \Pr_{\substack{F \from \Funcs_n\\R \from \Perms_n}} \left[\A^{R, F}(1^n) = 1\right]\right|} \\
%& = & 
\left| \Pr[\A(\Hyb_0)=1] - \Pr[\A(\Hyb_{\formerqC})=1] \right| 
%& \leq & \sum_{j=0}^{\formerqC-1} \left| \Pr[\A(\Hyb_j)=1]-\Pr[\A(\Hyb'_j)=1]\right| \\
%& & \mbox{} + \left| \sum_{j=1}^{\formerqC} \left( \Pr[\A(\Hyb'_{j-1})=1] - \Pr[\A(\Hyb_{j})=1] \right) \right| \\
& \leq & 1.5 \formerqC \sqrt{q_F/2^n}+ 2\sum_{j=1}^{\formerqC}q_{F,j}\sqrt{j/2^n} \\
& \leq & 1.5 \formerqC \sqrt{q_F/2^n}+ 2\sqrt{q_E/2^n}\sum_{j=1}^{\formerqC}q_{F,j} \\
& \leq & 1.5 \formerqC \sqrt{q_F/2^n} + 2q_F \sqrt{\formerqC/2^n}\,,
\end{eqnarray*}
as required.
%This concludes the proof of the theorem.
\qed\end{proof}

\section{Further Details for the Proof of Lemma~\ref{lem:step-one-perps}}

	%%%%%
%%%%%
%Put this figure here to prevent it from being placed in the inverse direction proof below, which would be confusing
%%%%%
% Let's just handle placement until we have resolved all comments
\begin{figure}[bht]
	%	\begin{center}
	\framebox[\textwidth][l]{
		%	\begin{minipage}{\textwidth}
		\begin{algorithm}[H]
			\DontPrintSemicolon
			\setcounter{AlgoLine}{22}
			$P, R \from \Perms_n$\;
			Run $\A$ with quantum access to $P$ and classical access to $R$,   until $\A$ makes its $(j+1)$st classical query~$x_{j+1}$;  let $T_j$ %:=\big((x_1,y_1), \cdots, (x_j,y_j)\big)$ is 
			be as in the text\;
			$s_0, s_1 \leftarrow \bool^n$\;
			$k_1 := s_0 \xor x_{j+1}$,  $k_2 \leftarrow D_{\mid k_1}$, $k:=(k_1, k_2)$ \;
			$y_{j+1} := P(s_1) \xor k_2$  \label{alg:1:compute-y} \; %=P(s_1) \xor k_2
			\lIf{$y_{j+1} \in \{y_1, \ldots, y_j\}$}{
				$y_{j+1} \leftarrow \bool^n\setminus\{y_1, \ldots, y_j\}$ }
			Give $y_{j+1}$ to $\A$ as the answer to its $(j+1)$st classical query\;
			$T_{j+1} := \big((x_1, y_1), \ldots, (x_{j+1}, y_{j+1})\big)$\;
			Continue running $\A$ with quantum access to $P_{T_{j+1},k}$ and classical access to~$E_k[P]$  \label{line:1:oracles} \;
		\end{algorithm}
		%\end{minipage}
	}
	\framebox[\textwidth][l]{
		%	\begin{minipage}{\textwidth}
		\begin{algorithm}[H]
			\DontPrintSemicolon
			\setcounter{AlgoLine}{31}
			$P, R \from \Perms_n$\;
			Run $\A$ with quantum access to $P$ and classical access to $R$,   until $\A$ makes its $(j+1)$st classical query~$x_{j+1}$;  let $T_j$ %:=\big((x_1,y_1), \cdots, (x_j,y_j)\big)$ is 
			be as in the text\;
			$k_1 \leftarrow \bool^n$,  $k_2 \leftarrow D_{\mid k_1}$, $k:=(k_1, k_2)$,
			$y_{j+1} \leftarrow \bool^n$  \; %=P(s_1) \xor k_2
			\lIf{$y_{j+1} \in \{y_1, \ldots, y_j\}$}{
				$y_{j+1} \leftarrow \bool^n\setminus\{y_1, \ldots, y_j\}$ }
			Give $y_{j+1}$ to $\A$ as the answer to its $(j+1)$st classical query\;
			$T_{j+1} := \big((x_1, y_1), \ldots, (x_{j+1}, y_{j+1})\big)$\;
			Continue running $\A$ with quantum access to $P_{T_{j+1},k}$ and classical access to~$E_k[P]$  \;
		\end{algorithm}
		%\end{minipage}
	}
	\framebox[\textwidth][l]{
		%	\begin{minipage}{\textwidth}
		\begin{algorithm}[H]
			\DontPrintSemicolon
			\setcounter{AlgoLine}{38}
			$P, R \from \Perms_n$\;
			Run $\A$ with quantum access to $P$ and classical access to $R$,   until $\A$ makes its $(j+1)$st classical query~$x_{j+1}$;  let $T_j$ %:=\big((x_1,y_1), \cdots, (x_j,y_j)\big)$ is 
			be as in the text\;
			%$k_1 \leftarrow \bool^n$,  $k_2 \leftarrow D_{\mid k_1}$, 
			$k\leftarrow D$, %=P(s_1) \xor k_2
			$y_{j+1} \leftarrow \bool^n\setminus\{y_1, \ldots, y_j\}$ \;
			Give $y_{j+1}$ to $\A$ as the answer to its $(j+1)$st classical query\;
			$T_{j+1} := \big((x_1, y_1), \ldots, (x_{j+1}, y_{j+1})\big)$\;
			Continue running $\A$ with quantum access to $P_{T_{j+1},k}$ and classical access to~$E_k[P]$  \;
		\end{algorithm}
		%\end{minipage}
	}
	\caption{Syntactic rewritings of ${\sf Expt}'_j$. \label{fig:1}} %\vspace*{-.5in}
	%	\end{center}
\end{figure}

%%%%%%%%%
%%%%%%%%%
%%%%%%%%%

\subsection{Equivalence of ${\sf Expt}'_j$ and $\Hyb'_j$}
\label{appendix:games}

The code in the top portion of Figure~\ref{fig:1} is a syntactic rewriting of~${\sf Expt}'_j$. (Flags that have no effect on the output of $\A$ are omitted.) In line~\ref{alg:1:compute-y}, the computation of $y_{j+1}$  has been expanded (note that $E_k[P_1](x_{j+1}) = P_1(s_0) \xor k_2 = P(s_1) \xor k_2$).
In line~\ref{line:1:oracles}, $Q$ has been replaced with $P_{T_{j+1},k}$ and $\O$ has been replaced with $E_k[P]$ as justified in the proof of Lemma~\ref{lem:step-one-perps}.

The code in the middle portion of Figure~\ref{fig:1} results from the following changes: first, rather than sampling uniform $s_0$ and then setting $k_1:=s_0 \xor x_{j+1}$, the code now samples a uniform~$k_1$. 
Similarly, rather than choosing uniform $s_1$ and then setting $y_{j+1}:=P(s_1) \xor k_2$, the code now samples a uniform~$y_{j+1}$ (note that $P$ is a permutation, so $P(s_1)$ is uniform).
Since neither $s_0$ nor $s_1$ is used anywhere else, each can now be omitted. 

The code in the bottom portion of Figure~\ref{fig:1} simply chooses $k=(k_1, k_2)$ according to distribution~$D$, and chooses uniform $y_{j+1} \in \bool^n \setminus \{y_1, \ldots, y_j\}$.
It can be verified by inspection that this final experiment is equivalent to~$\Hyb'_j$.

\subsection{Handling an Inverse Query}
\label{appendix:inverse-case}
In this section we discuss the case where the \mbox{$(j+1)$st} classical query of $\A$ is a inverse query in the proof of \expref{Lemma}{lem:step-one-perps}. Phase~1 is exactly as described in the proof of \expref{Lemma}{lem:step-one-perps}, though we now let $y_{j+1}$ denote the \mbox{$(j+1)$st} classical query made by~$\A$ (assumed to be in the inverse direction).
\begin{description}
\item[Phase 2:] 
$\D$ receives $s_0, s_1 \in \bool^n$ and quantum oracle access to a permutation~$P_b$. First $\D$ sets $t_0:=P_b(s_0)$ and $t_1:=P_b(s_1)$. It then %Note that $t_0, t_1$ are uniform in $\bool^n$.
sets $k_2:=t_0 \xor y_{j+1}$, %=P_b(s_0)\xor y_{j+1}$, 
chooses $k_1 \leftarrow D_{|k_2}$ (where this represents the conditional distribution on $k_1$ given~$k_2$), and
sets $k:=(k_1, k_2)$. $\D$ continues running $\A$, answering its remaining classical queries (including the $(j+1)$st one) using $E_k[P_b]$, and its remaining quantum queries using 
\[
%[P_b]_{T, k}^\pm = 
(P_b)_{T_j, k} = \swap_{P_b(x_1 \oplus k_1), y_1 \oplus k_2} \circ \cdots \circ \swap_{P_b(x_j \oplus k_1), y_j \oplus k_2} \circ P_b\,.
%P_b\circ\swap_{(x_j \oplus k_1,P^{-1}( y_j \oplus k_2))} \circ \cdots \circ \swap_{(x_1 \oplus k_1,P^{-1}( y_1 \oplus k_2))}\,.
\]
Finally, $\D$ outputs whatever $\A$ outputs.
\end{description}

Note that $t_0, t_1$ are uniform, and so $k$ is distributed according to~$D$.
Then:

\medskip\noindent \textbf{Case $b=0$ (no reprogramming).} 
In this case, $\A$'s remaining classical queries (including its $(j+1)$st classical query) are answered using $E_k[P_0] = E_k[P]$, and its remaining quantum queries are answered using $(P_0)_{T_j, k} = P_{T_j, k}$. The output of $\A$ is thus distributed identically to its output in~$\Hyb_j$ in this case.

\medskip\noindent \textbf{Case $b=1$ (reprogramming).} 
In this case, $k_2=P_1(s_0) \xor y_{j+1}=P(s_1)\xor y_{j+1}$ and so 
%\begin{eqnarray*}
%P_1 = P \circ \swap_{s_0, s_1} = \swap_{P(s_0), P(s_1)} \circ P = \swap_{P(s_0), y_{j+1}\oplus k_2} \circ P\,.
%\end{eqnarray*}
%and  
\begin{eqnarray*}
P_b^{-1} = P_1^{-1} = (P \circ \swap_{s_0,s_1})^{-1} & = & (\swap_{P(s_0), P(s_1)} \circ P)^{-1} \\
& = & P^{-1} \circ \swap_{P(s_0), P(s_1)} \\
& = & P^{-1} \circ \swap_{P(s_0), y_{j+1}\oplus k_2} .
\end{eqnarray*}
The response to $\A$'s $(j+1)$st classical query  is thus
\begin{eqnarray*}
x_{j+1} \stackrel{\rm def}{=} E_k^{-1}[P_1](y_{j+1}) = P_1^{-1}(y_{j+1} \oplus k_2) \oplus k_1 = P_1^{-1}(P(s_1)) \xor k_1 = s_0 \oplus k_1\,.
\end{eqnarray*}
The remaining classical queries of $\A$ are then answered using $E_k[P_1]$, while its remaining quantum queries are answered using~$(P_1)_{T_j, k}$. 

Now we define the following three events: 
\begin{enumerate}
		\item $\qcoll$ is the event that $x_{j+1} \in \{x_1, \ldots, x_j\}$.
\item $\scoll$  is the event that $P(s_0) \xor k_2 \in \{y_1, \ldots, y_j\}$. %\jnote{Note that $x_i = s_0 \xor k_1$ is disallowed by assumption.}
\item $\find$  is the event that, in phase~2, $\A$ queries its classical oracle in the forward direction on $s_1 \xor k_1$, %=s_0\oplus s_1\oplus x^*$, 
or the inverse direction on $P(s_0) \xor k_2$. 
\end{enumerate}

Comparing the above to the proof of Lemma~\ref{lem:step-one-perps}, we see (because $P$ is a permutation) that the situation is entirely symmetric, and the analysis is therefore the same.

\end{document}